\documentclass[lettersize,journal]{IEEEtran}
\usepackage{amsmath,amsfonts,amssymb,amsthm}
\usepackage{bm}
\usepackage{cite}
\usepackage{mathtools}
\usepackage{xcolor}
\usepackage{algorithmic}
\usepackage[ruled]{algorithm2e}
\usepackage{array}
\usepackage{textcomp}
\usepackage{stfloats}
\usepackage{url}
\usepackage{verbatim}
\usepackage{graphicx}
\def\BibTeX{{\rm B\kern-.05em{\sc i\kern-.025em b}\kern-.08em
    T\kern-.1667em\lower.7ex\hbox{E}\kern-.125emX}}
\usepackage{balance}
\usepackage{tabularx,booktabs}
\usepackage{apxproof}
\usepackage{multicol}
\usepackage{subcaption}

\theoremstyle{plain}% Theorem-like structures provided by amsthm.sty
\newtheorem{theorem}{Theorem}
\newtheorem{lemma}{Lemma}
\newtheorem{corollary}{Corollary}
\newtheorem{proposition}{Proposition}

\theoremstyle{definition}
\newtheorem{definition}{Definition}
\newtheorem{example}{Example}
\newtheorem{assumption}{Assumption}
\newtheorem{remark}{Remark}
\newtheorem{fact}{Fact}

\theoremstyle{remark}

\begin{document}
% symbols
\newcommand{\Real}{\mathbb{R}}
\newcommand{\Natural}{\mathbb{N}}
\newcommand{\Bdd}{\mathcal{B}}
\newcommand{\Xsp}{\mathcal{X}}
\newcommand{\Ysp}{\mathcal{Y}}
\newcommand{\Zsp}{\mathcal{Z}}
\newcommand{\Pcal}{\mathcal{P}}

\newcommand{\Am}{\bm{A}}
\newcommand{\Bm}{\bm{B}}
\newcommand{\Dm}{\bm{D}}
\newcommand{\Em}{\bm{E}}
\newcommand{\Lm}{\bm{L}}
\newcommand{\Mm}{\bm{M}}
\newcommand{\Xm}{\bm{X}}
\newcommand{\Ym}{\bm{Y}}
\newcommand{\Zm}{\bm{Z}}
\newcommand{\Nm}{\bm{N}}
\newcommand{\Xim}{\bm{\Xi}}
\newcommand{\Imat}{\bm{I}}

\newcommand{\av}{\bm{a}}
\newcommand{\bv}{\bm{b}}
\newcommand{\ev}{\bm{e}}
\newcommand{\xv}{\bm{x}}
\newcommand{\sv}{\bm{s}}
\newcommand{\uv}{\bm{u}}
\newcommand{\vv}{\bm{v}}
\newcommand{\zv}{\bm{z}}
\newcommand{\yv}{\bm{y}}
% parentheses, brackets and braces
\newcommand{\parentheses}[1]{\left( #1 \right )}
\newcommand{\brackets}[1]{\left[ #1 \right]}
\newcommand{\braces}[1]{\left\{ #1 \right\}}
% operations
\newcommand{\norm}[1]{\left\| #1 \right\|}
\newcommand{\argmin}{\mathrm{arg\,min}}
\newcommand{\innerp}[2]{\langle #1, #2 \rangle}
\newcommand{\lev}[1]{\textrm{lev}_{\leq #1}}
\newcommand{\dom}{\textrm{dom\,}}
\newcommand{\cont}{\textrm{cont\,}}
\newcommand{\interior}{\textrm{int\,}}
\newcommand{\txtred}[1]{\textcolor{red}{#1}}
\newcommand{\txtblue}[1]{\textcolor{blue}{#1}}
\newcommand{\Prox}{\mathrm{Prox}}
\newcommand{\Conv}{\textrm{Conv}}
\newcommand{\dist}{\textrm{dist}}
\newcommand{\levelset}{\mathrm{lev}}
\newcommand{\blkdiag}{\textrm{blkdiag}}
\newcommand{\vectorize}{\textrm{vec}}
\newcommand{\trace}{\textrm{tr}}

\title{A Unified Framework for Solving a General Class of Nonconvexly Regularized Convex Models}
\author{Yi Zhang and Isao Yamada, \IEEEmembership{Fellow, IEEE}
\thanks{Manuscript created January, 2023; This work was supported in part by JSPS Grants-in-Aid (21J22393), JSPS Grants-in-Aid (19H04134) and by JST SICORP (JPMJSC20C6).}
\thanks{The authors are with the Department of Information and Communications Engineering, Tokyo Institute of Technology, Meguro-ku, Tokyo 152-8550, Japan (email: yizhang@sp.ict.e.titech.ac.jp; isao@sp.ict.e.titech.ac.jp)}}

%\markboth{Journal of \LaTeX\ Class Files}%
%{How to Use the IEEEtran \LaTeX \ Templates}

\maketitle

\begin{abstract}
Recently, several nonconvex sparse regularizers which can preserve the convexity of the cost function have received increasing attention. This paper proposes a general class of such convexity-preserving (CP) regularizers, termed partially smoothed difference-of-convex (pSDC) regularizer. The pSDC regularizer is formulated as a structured difference-of-convex (DC) function, where the landscape of the subtrahend function can be adjusted by a parameterized smoothing function so as to attain overall-convexity. Assigned with proper building blocks, the pSDC regularizer reproduces existing CP regularizers and opens the way to a large number of promising new ones. 

With respect to the resultant nonconvexly regularized convex (NRC) model, we derive a series of overall-convexity conditions which naturally embrace the conditions in previous works. Moreover, we develop a unified framework based on DC programming for solving the NRC model. Compared to previously reported proximal splitting type approaches, the proposed framework makes less stringent assumptions. We establish the convergence of the proposed framework to a global minimizer. Numerical experiments demonstrate the power of the pSDC regularizers and the efficiency of the proposed DC algorithm.
\end{abstract}

\begin{IEEEkeywords}
sparse recovery, nonconvex penalties, convexity-preserving regularizers, nonconvexly regularized convex models, DC programming
\end{IEEEkeywords}

\section{Introduction}
\IEEEPARstart{T}{he} last two decades have witnessed a flourish of studies in sparsity-aware processing \cite{candes2006,donoho2006,bach2012,eldar2012}. A prominent approach to estimate sparse signals is the variational method which pursues a minimizer of the following cost function: 
\begin{equation}\label{eq:sparse_recovery}
\underset{\xv\in\Real^n}{\text{minimize}}\;\;J(\xv)\coloneqq F(\xv)+\lambda \Psi(\xv),
\end{equation}
where $F:\Real^n\to\Real$ is the data fidelity term, $\lambda>0$ is a tuning parameter, and $\Psi:\Real^n\to\Real\cup\{+\infty\}$ is a regularizer that promotes sparseness of the solution.

The ideal choice for $\Psi$ is the $l_0$ pseudo-norm, i.e., the number of nonzero components in the input vector. However, this leads to a discontinuous nonconvex optimization problem which is known to be NP-hard \cite{natarajan1995}. Accordingly, one usually resorts to continuous approximations of $l_0$ pseudo-norm to circumvent this difficulty. Conventional studies usually have adopted convex regularizers (e.g., $l_1$-norm \cite{tibshirani1996}, Huber function \cite{fessler1997}), which ensure efficient and reliable solution of (\ref{eq:sparse_recovery}). Nevertheless, since most convex regularizers are coercive (i.e., $\Psi(\xv)$ goes to $+\infty$ if $\norm{\xv}_2$ goes to $+\infty$), they usually overpenalize the $i$th component $x_i$ when $|x_i|$ is large, which causes underestimation of the true solution \cite{zhang2010}. To overcome this problem, continuous nonconvex regularizers (e.g., SCAD \cite{fan2001}, MCP \cite{zhang2010}) have been proposed to accomplish less biased estimation. Although they yield much more tractable nonconvex programs than the original $l_0$ pseudo-norm, when adopted in application, existing algorithms may get stuck in local minima, which poses a concern on the reliability of nonconvex models.

To resolve this dilemma, particular convexity-preserving (CP) regularizers have been proposed to achieve debiased convex sparse regularization (see Sec. \ref{sec:related_works} for details). The so-called CP regularizer is a special parameterized regularizer. Although the CP regularizer itself is nonconvex, its shape can be adjusted by certain tuning parameter so as to induce the overall-convexity of the cost function\footnote{More precisely, given any data fidelity term $F$ and weight parameter $\lambda$ satisfying suitable conditions, one can always find some proper tuning parameter value so that the resultant cost function is convex.}. Therefore, CP regularizers enjoy both improved estimation accuracy and reliability of solution. Due to such favourable properties, there is an increasing need to design CP regularizers for more general regularization problems than sparse regularization, and to develop unified solution methods for the resultant nonconvexly regularized convex (NRC) models, which constitutes the motivation of the current study.

Earlier CP regularizers are mostly case-specifically designed for the sparse least-squares problem \cite{selesnick2017}, the low rank matrix recovery problem \cite{parekh2016} or particular variants of them \cite{selesnick2020,suzuki2021}. To further exploit the power of CP regularizers, more general formulations of CP regularizers \cite{abe2020,alshabili2021} have been proposed, and solution algorithms \cite{yata2021} for NRC models with split feasibility type constraints \cite{censor2005} have been developed. However, the existing studies have the following limitations:
\begin{enumerate}
\item Previous CP regularizers \cite{selesnick2017,abe2020,alshabili2021} mainly consider quadratic data fidelity terms, which do not cover general observation systems such as computed tomography \cite{bouman1996}.

\item Existing proximal splitting type algorithms  \cite{selesnick2017,abe2020,alshabili2021,yata2021} for NRC models rely heavily on the proximability\footnote{See Sec. \ref{subsec:notation} for the definition of "proximable" functions.} of the component functions, which may fail in practice.

\item Existing proximal splitting type algorithms \cite{selesnick2017,abe2020,alshabili2021,yata2021} for NRC models deal with multiple regularizers \cite[Example 3]{abe2020} and convex constraint \cite{yata2021} by lifting the problem to a higher-dimensional space.
\end{enumerate}

The current paper is devoted to generalizing prior arts as well as tackling the difficulties above. Our major contributions can be summarized as follows:
\begin{enumerate}
\item We propose a general class of CP regularizers termed partially smoothed difference-of-convex (pSDC) regularizer. The pSDC regularizer enjoys remarkable power of representability.  Assigned with proper building block functions, the pSDC regularizer reproduces existing CP regularizers \cite{selesnick2017,abe2020,alshabili2021} and opens the way to a large number of promising new ones. 

\item We study an abstract NRC model which takes general convex data fidelity terms, general convex constraint and multiple regularizers into account. With respect to this NRC model, we derive a series of overall-convexity conditions which naturally embraces the conditions proposed in previous works \cite{parekh2016,selesnick2017,abe2020,alshabili2021,suzuki2021}.

\item We develop a unified framework for solving the proposed NRC model. This DC programming \cite{le2018} based framework does not assume component functions to be proximable, thus is less stringent. Moreover, it does not necessarily enlarge the dimension of the problem when additional regularizers and constraints are involved.
\end{enumerate}

The remainder of this paper is organized as follows. Section \ref{sec:preliminaries} introduces the necessary mathematical preliminaries and provides a brief introduction on the conventional studies of CP regularizers. In Section \ref{sec:pSDC_regularizer}, we formally propose the pSDC regularizer, and demonstrate its powerful representability by concrete examples. In Section \ref{sec:overall_convexity_conditions}, we study an abstract NRC model which incorporates multiple pSDC regularizers along with general convex data fidelity terms and general convex constraint, and we present a series of overall-convexity conditions with respect to this NRC model. In Section \ref{sec:dc_algorithm}, we develop a DC type solution algorithm for the proposed NRC model and establish its convergence to a global minimizer. Section \ref{sec:numerical_experiments} provides results of numerical experiments, followed by conclusion in Section \ref{sec:conclusion}. A preliminary short version of this paper was presented at a conference \cite{zhang2022}.

\section{Preliminaries}\label{sec:preliminaries}
\subsection{Mathematical Preliminaries}\label{subsec:notation}
Let $\Natural,\Real,\Real_+$ be the sets of nonnegative integers, real numbers and positive real numbers. For $n$-dimensional Euclidean space $\Real^n$, $\innerp{\cdot}{\cdot}$ and $\norm{\cdot}_{p}\;(p\geq 1)$ denote respectively the standard inner product and the $l_p$-norm in $\Real^n$. $\mathbf{0}_n$ stands for the $n\times 1$ zero vector and $\mathbf{O}_{m\times n}$ denotes the $m\times n$ zero matrix. $\Imat_n$ denotes the $n\times n$ identity matrix. $\mathrm{diag}\parentheses{d_1,d_2,\dots,d_n}$ denotes the diagonal matrix with diagonal entries $d_1,d_2,\dots,d_n$. For $\Am\in\Real^{m\times n}$, $\Am^T\in\Real^{n\times m}$ denotes its transpose. For a square matrix $\Am\in\Real^{n\times n}$, $\trace(\Am)$ denotes its trace. For $\Am\in\Real^{m\times n}$, $\norm{\Am}_F$ and $\norm{\Am}_2$ respectively denote its Frobenius norm and spectral norm, and $\norm{\Am}_{2,1}\coloneqq \sum_{j=1}^n\norm{\av_j}_2$ is the $l_{2,1}$-norm, where $\av_j$ is the $j$th column vector of $\Am$.

For $f:\Real^n\to[-\infty,+\infty]$, the lower level set of $f$ at height $\xi\in\Real$ is defined as $\levelset_{\leq \xi}f\coloneqq\braces{\xv\in\Real^n\mid f(\xv)\leq \xi}$, and the domain of $f$ is $\dom{f}\coloneqq\braces{\xv\in\Real^n\mid f(\xv)<+\infty}$. For a differentiable function $f:\Real^{n}\to\Real$, $\nabla f(\xv)$ denotes its gradient at $\xv\in\Real^n$. For a twice differentiable function $f:\Real^n\to\Real$, $\nabla^2 f(\xv)\in\Real^{n\times n}$ denotes its Hessian matrix at $\xv\in\Real^n$. For a convex function $f:\Real^n\to\Real\cup\braces{+\infty}$, the subdifferential of $f$ at $\xv\in\Real^n$ is
\begin{equation*}
\partial f(\xv)\coloneqq\{\uv\in\Real^n\mid(\forall \zv\in\Real^n)\;\innerp{\zv-\xv}{\uv}+f(\xv)\leq f(\zv)\},
\end{equation*}
if $\uv\in\partial f(\xv)\neq\emptyset$, then $\uv$ is called a subgradient of $f$ at $\xv$. We denote $\Gamma_0(\Real^n)$ as the set of all proper lower semicontinuous\footnote{For $f:\Real^N\to\Real\cup\{+\infty\}$, $f$ is referred to as \textit{proper} if $\dom{f}\coloneqq \braces{\xv\in\Real^N\mid f(\xv)<+\infty}\neq \emptyset$, and \textit{lower semicontinuous} if $\mathrm{lev}_{\leq \xi}f$ is closed for every $\xi\in\Real$ \cite[Thm. 2.6]{beck2017}.} convex functions from $\Real^n$ to $\Real\cup\braces{+\infty}$ \cite{bauschke2017}. For $f$ in $\Gamma_0(\Real^n)$, the proximity operator of $f$ is defined as
\begin{equation*}
\Prox_{f}(\xv)=\underset{\zv\in\Real^n}{\argmin}\; \brackets{f(\zv)+\frac{1}{2}\norm{\xv-\zv}_2^2}.
\end{equation*}
We say that $f$ is proximable if $\Prox_{\gamma f}$ can be computed to high precision efficiently for every $\gamma>0$. For mappings $F_1:\Real^n\to\Real^m$ and $F_2:\Real^m\to\Real^p$, $F_2\circ F_1$ denotes their composition. For a nonempty closed convex set $C\subset \Real^n$, the indicator function of $C$ is defined as
\begin{equation*}
\iota_C:\Real^n\to[-\infty,+\infty]:x\mapsto\begin{cases}
0, & \mathrm{if }\;\; x\in C,\\
+\infty, & \mathrm{otherwise.}
\end{cases}
\end{equation*}
For $f,g:\Real^n\to\Real\cup\{+\infty\}$, $(f\Box g)$ stands for the infimal convolution \cite{bauschke2017} of them:
\begin{equation}
(f\Box g)(\xv)\coloneqq \inf_{\zv\in\Real^n} \parentheses{f(\zv)+g(\xv-\zv)},
\end{equation}
and the infimal convolution is \textit{exact at} $\xv\in\Real^n$ if
\begin{equation*}
(\exists \zv\in\Real^n)\;\; (f\Box g)(\xv)=f(\zv)+g(\xv-\zv)\in (-\infty,+\infty];
\end{equation*}
$(f\Box g)$ is \textit{exact} if it is exact at every point of its domain.

In the sequel, we introduce the inf-smoothing technique which play a critical role in the proposed pSDC regularizer.
\begin{definition}\label{def:inf_conv_approximation}
Let $f\in\Gamma_0(\Real^n)$ and let $\omega:\Real^n\to\Real$ be a differentiable convex function with $\nabla\omega$ being Lipschitz continuous with constant $1/\sigma\; (\sigma>0)$. Suppose that $(f\Box \omega)(\xv)$ is finite for every $\xv\in\Real^n$. Then we call $f\Box \omega$ \textit{the inf-conv smooth approximation of $f$ by $\omega$}.
\end{definition}
The definition above is a slightly modified version of \cite[Def. 4.2]{beck2012}. It implies that given a convex function $f$ and a smooth convex function $\omega$, $f\Box\omega$ is a smooth approximation of $f$, as summarized in the following proposition.
\begin{proposition}\label{prop:inf_conv_smooth_approximation}
Let $f\Box\omega$ be the inf-conv smooth approximation of $f$ by $\omega$. Then the following holds:
\begin{enumerate}
\item $f\Box \omega$ is differentiable and with gradient $\nabla (f\Box\omega)$ which is Lipschitz with constant $\frac{1}{\sigma}$,

\item let $\xv\in\Real^n$, and suppose that $\zv_{\xv}$ satisfies that
\begin{equation*}
\zv_{\xv}\in\underset{{\zv\in\Real^n}}{\argmin}\; f(\zv)+\omega(\xv-\zv),
\end{equation*}
then $\nabla (f\Box\omega)(\xv)=\nabla\omega(\xv-\zv_{\xv})$,

\item $(f\Box \omega)\in\Gamma_0(\Real^n)$.
\end{enumerate}
\end{proposition}
\begin{proof}
The results 1) and 2) follows from \cite[Thm. 4.1]{beck2012}. Combining the finite-valuedness, continuity of $(f\Box \omega)$ with \cite[Prop. 12.11]{bauschke2017} yields 3).
\end{proof}

From Proposition \ref{prop:inf_conv_smooth_approximation}, one can imagine that $f\Box\omega$ is a smooth convex approximation of $f$, and its shape can be adjusted by the smoothing function $\omega$. As will be shown in Section \ref{sec:pSDC_regularizer} and \ref{sec:overall_convexity_conditions}, by equipping $\omega$ with a shape-controlling tuning parameter $\Pcal$, a part of the pSDC regularizer would be deformable. And when the pSDC regularizer is adopted as $\Psi$ in (\ref{eq:sparse_recovery}), one can adjust $\Pcal$ to attain the overall-convexity of $J$ in (\ref{eq:sparse_recovery}).

\subsection{Related Works}\label{sec:related_works}
\subsubsection{Primitive CP Regularizers}

The idea of CP regularizers and NRC models dates back to over three decades ago \cite{blake1987,nikolova1998,nikolova1999}. However, earlier studies \cite{blake1987,nikolova1998,nikolova1999,parekh2016,suzuki2021} usually assume the presence of a strongly convex term (e.g., a strongly convex data fidelity or the $l_2$ regularization term), hence are fundamentally limited. For example, in \cite{nikolova1998}, the author proposed the following nonconvex regularizer:
\begin{equation}\label{eq:CP_regularizer_for_binary_image}
\Psi_{\mathrm{BI}}\parentheses{\xv;\alpha}= \sum_{i\sim j}\beta_{i,j}\lvert x_i-x_j \rvert-\alpha\sum_{i}\parentheses{x_i-\frac{1}{2}}^2
\end{equation}
for estimating binary images\footnote{One can verify that the first term of $\Psi_{\mathrm{BI}}$ is a convex term which promotes the correlated structure of the image, and the second term is a concave term which promotes the binarity of pixels.}, where $\xv\coloneqq\begin{bmatrix}x_1,\dots,x_n\end{bmatrix}^T\in\Real^n$ is the estimate of the unknown binary image, $i\sim j$ means that $x_i$ and $x_j$ are neighbouring pixels, $\bm{\beta}\coloneqq \parentheses{\beta_{i,j}}_{i\sim j}\subset \Real_+$, $\alpha>0$ is the shape controlling tuning parameter. 

Since the first term in (\ref{eq:CP_regularizer_for_binary_image}) is convex and the second is concave, $\Psi_{\mathrm{BI}}$ is the difference between two convex functions, i.e., $\Psi_{\mathrm{BI}}$ is a difference-of-convex (DC \cite{le2018}) function. We consider the quadratic data fidelity $F_{\mathrm{quad}}(\xv)\coloneqq \frac{1}{2}\norm{\yv-\Am\xv}^2_2$ and the following cost function
\[J_{\mathrm{BI}}(\xv;\alpha)\coloneqq F_{\mathrm{quad}}(\xv)+\lambda\Psi_{\mathrm{BI}}(\xv;\alpha).\]
It has been proved in \cite{nikolova1998} that if $\alpha$ satisfies\footnote{$\lambda_{\min}(\cdot)$ is the smallest eigenvalue of the input matrix.}
\begin{equation}\label{eq:convexity_condition_BI}
\lambda_{\min}(\Am^T\Am)\geq \lambda\alpha>0,
\end{equation}
then the concave second term of $\lambda\Psi_{\mathrm{BI}}$ would be overpowered by $F_{\mathrm{quad}}$, whereby the cost function $J_{\mathrm{BI}}(\cdot;\alpha)$ is convex.

This example reveals an important fact, that is: \textit{if there exists a strongly convex term in the cost function $J$, then we can introduce some concave terms into the regularizer $\Psi$ to improve its regularizing properties while maintaining the overall-convexity of the cost function}. However, (\ref{eq:convexity_condition_BI}) implies the nonsingularity of $\Am^T\Am$, which usually fails to hold in applications such as sparse recovery problems.

\subsubsection{The Generalized Minimax Concave Penalty}
\label{sec:GMC}
The first CP regularizer that does not require strong convexity of $F_{\mathrm{quad}}$ (i.e., nonsingularity of $\Am^T\Am$) is the generalized minimax concave (GMC) penalty \cite{selesnick2017} defined as follows:
\begin{equation}
\Psi_{\text{GMC}}(\xv;\Bm)\coloneqq l_1(\xv)-(l_1\Box q_{\Bm})(\xv),
\end{equation}
where $l_1(\xv)\coloneqq \norm{\xv}_1$ is the $l_1$-norm, $q_{\Bm}(\xv)\coloneqq \frac{1}{2}\norm{\Bm\xv}^2_2$ is a quadratic smoothing function with $\Bm\in\Real^{p\times n}$ being the shape controlling tuning parameter. The GMC penalty is the difference between the $l_1$-norm and its inf-conv smooth approximation by $q_{\Bm}$, hence as $\Psi_{\mathrm{BI}}$, $\Psi_{\mathrm{GMC}}$ is a DC function. 

We note that the GMC penalty is a nonseparable multidimensional generalization of the minimax concave penalty (MCP \cite{zhang2010}), more precisely, if $\Bm^T\Bm$ is diagonal, then $\Psi_{\mathrm{GMC}}$ reproduces a weighted sum of MCP, which accounts for the name of the GMC penalty. In contrast to the standard MCP, the shape of $\Psi_{\text{GMC}}(\cdot;\Bm)$ can be adjusted flexibly via changing $\Bm$. For the following cost function
\[J_{\mathrm{GMC}}(\xv;\Bm)\coloneqq F_{\mathrm{quad}}(\xv)+\lambda\Psi_{\mathrm{GMC}}(\xv;\Bm),\]
it is proved in \cite{selesnick2017} that if $\Bm$ satisfies
\begin{equation}\label{eq:convexity_condition_GMC}
\Am^T \Am\succeq \lambda \Bm^T \Bm,
\end{equation}
then the concave term $-\lambda(l_1\Box q_{\Bm})(\xv)$ is overpowered by $F_{\mathrm{quad}}(\xv)$, and the cost function $J_{\mathrm{GMC}}(\cdot;\Bm)$ is convex. 

Remarkably, (\ref{eq:convexity_condition_GMC}) does not require $\Am^T\Am$ to be nonsingular as (\ref{eq:convexity_condition_BI}) does. Instead, by (\ref{eq:convexity_condition_GMC}), $\Psi_{\mathrm{GMC}}$ is able to exploit the "partially strong convexity" of $F_{\mathrm{quad}}$, more precisely, \textit{if the data fidelity term is strongly convex in certain direction (e.g. the eigenvectors of $\Am^T\Am$ with nonzero eigenvalues), then one can introduce concavity into $\Psi_{\mathrm{GMC}}$ in this direction to achieve better approximation of the $l_0$ pseudo-norm}.

\noindent
\subsubsection{Extensions of the GMC Penalty}
Certain efforts have been made to broaden applicability of $\Psi_{\mathrm{GMC}}$. One notable extension is the linearly involved generalized Moreau enhanced (LiGME) model \cite{abe2019,abe2020}:
\begin{equation}
\Psi_{\text{LiGME}}(\xv;\Bm)= \psi(\Lm \xv)-(\psi\Box q_{\Bm})(\Lm \xv),
\end{equation}
where $\Lm\in\Real^{q\times n}$ is the analysis matrix which encodes the sparsifying domain of the interested signal; $\psi\in\Gamma_0(\Real^q)$ is a kernel function which is no longer restricted to the $l_1$-norm, but can be any proximable function. Accordingly, the LiGME model allows applying the construction technique of GMC to more general convex kernel functions. A variant of the LiGME model with split feasibility type constraints is studied in \cite{yata2021}. 

Another useful extension of GMC is the sharpening sparse regularizers (SSR) framework \cite{alshabili2021}:
\begin{equation}
\Psi_{\text{SSR}}(\xv;\Bm)\coloneqq l_1(\xv)-((l_1\circ \Lm)\Box (\Phi\circ \Bm))(\xv),
\end{equation}
where $\Lm\in\Real^{q\times n}$ is the analysis matrix which is embedded at a different position from the LiGME model, $\Phi(\zv)\coloneqq \sum_{i=1}^q \phi(z_i)$ with $\phi\in\Gamma_0(\Real)$ is an isotropic smoothing function which is not restricted to the $l_2$-norm. While the SSR model does not consider variability of the kernel function, it allows adopting different smoothing function $\Phi$, thus can adjust the shape of the regularizer more delicately. 

So far, overall-convexity conditions and proximal splitting type \cite{combettes2011} algorithms have been developed independently with respect to the GMC \cite{selesnick2017}, LiGME \cite{abe2020} and SSR \cite{alshabili2021} models.

\renewcommand{\arraystretch}{1.5}
\begin{table*}
\centering
\caption{Prior arts as special instances of the pSDC regularizers}
\label{table:pSDC_reproduce_prior_arts}
\begin{tabular}{|c|c|c|c|c|c|}
\hline
&  $\psi_1(\xv)$ & $\psi_2(\zv)$ & $\phi_{\Pcal}(\zv)$ & $\Mm$ & $\Pcal$\\
\hline
$\Psi_{\mathrm{BI}}(\cdot;\alpha)$ & $\sum_{i\sim j}\beta_{i,j}\lvert x_i-x_j \rvert$ & $\sum_{i=1}^q \iota_{\braces{1/2}}(z_i)$ & $\alpha\norm{\zv}^2_2$ & $\Imat_n$ & $\alpha$\\
\hline
$\Psi_{\mathrm{GMC}}(\cdot;\Bm)$ & $\norm{\xv}_1$ & $\norm{\zv}_1$  & $\frac{1}{2}\norm{\Bm\zv}^2_2$ & $\Imat_n$ & $\Bm$\\
\hline
$\Psi_{\mathrm{LiGME}}(\cdot;\Bm)$ & $\psi(\Lm\xv)$ & $\psi(\zv)$ &  $\frac{1}{2}\norm{\Bm\zv}^2_2$ & $\Lm$ & $\Bm$\\
\hline
$\Psi_{\mathrm{SSR}}(\cdot;\Bm)$ & $\norm{\xv}_1$ & $\norm{\Lm\zv}_1$ & $\Phi(\Bm\zv)$ & $\Imat_n$ & $\Bm$\\
\hline
\end{tabular}
\end{table*}

\section{The Partially Smoothed Difference-of-Convex Regularizers}
\label{sec:pSDC_regularizer}
\subsection{Abstract Formulation}
The proposed pSDC regularizer is formulated as follows\footnote{The formulation of the pSDC regularizer was previously reported in \cite{zhang2022}. In this paper we newly add Example \ref{exp:capped_l1} to illustrate the applicability of the pSDC regularizer to more general regularization problems. Besides, the overall-convexity conditions discussed in Section \ref{sec:overall_convexity_conditions} are more general.}:
\begin{equation}\label{eq:cp_regularizer}
\Psi_{{p}}(\xv;\mathcal{P})\coloneqq \psi_1(\xv)-(\psi_2\Box\phi_{\mathcal{P}})(\Mm \xv),
\end{equation}
where $\psi_1\in\Gamma_0(\Real^n)$ and $\psi_2\in\Gamma_0(\Real^q)$ are kernel functions, $\Mm\in\Real^{q\times n}$ is an analysis matrix, $\Pcal$ is the shape controlling tuning parameter (which may be a number, a vector or a matrix) and the corresponding smoothing function $\phi_{\Pcal}\in\Gamma_0(\Real^q)$ is differentiable with $\nabla\phi_{\Pcal}$ being Lipschitz continuous on $\Real^q$ with constant $1/\sigma_{\Pcal}$. In contrast to the LiGME and SSR models introduced above, the pSDC regularizer permits at the same time the presence of analysis matrix, the variability of both the kernel functions and the smoothing function, thus its construction is very flexible (see Table \ref{table:pSDC_reproduce_prior_arts}).

The pSDC regularizer is the difference between $\psi_1(\cdot)$ and $(\psi_2\Box\phi_{\Pcal})(\Mm\cdot)$, where the former is a kernel function and the latter is a linearly involved smooth approximation of another kernel function. By Proposition \ref{prop:inf_conv_smooth_approximation}, both $\psi_1(\cdot)$ and $(\psi_2\Box\phi_{\Pcal})(\Mm\cdot)$ are convex functions, hence similar to existing CP regularizers introduced in Section \ref{sec:related_works}, $\Psi_{{p}}(\cdot;\Pcal)$ is a parameterized DC function whose subtrahend part can be adjusted by the tuning parameter $\Pcal$.

In applications, we expect $\psi_1,\psi_2$ to determine the basic landscape of the pSDC regularizer such that $\Psi_{{p}}(\cdot;\Pcal)$ serves as a qualified regularizer (which is the reason we call $\psi_1,\psi_2$ kernel functions). On the other hand, we adjust the shape of $\Psi_{{p}}(\cdot;\Pcal)$ by the smoothing function $\phi_{\Pcal}$ to attain overall-convexity.
More precisely, if $\Pcal$ satisfies the overall-convexity conditions specified in Section \ref{sec:overall_convexity_conditions}, the cost function in (\ref{eq:sparse_recovery}) is convex despite nonconvexity of $\Psi_{{p}}(\cdot;\Pcal)$.

We note that the pSDC regularizer is a unification and generalization of prior arts (cf. Example \ref{exp:existing_CP_sparse_regularizers}). In the sequel, we demonstrate the powerful representability of pSDC regularizers by presenting concrete examples. The convexity-preserving property of the pSDC regularizer will be studied in Section \ref{sec:overall_convexity_conditions}.

\subsection{Concrete Examples of the pSDC Regularizers}

First, we show that all of the existing CP regularizers introduced in Section \ref{sec:related_works} can be regarded as special instances of the pSDC regularizer. See the following example.

\begin{example}[Existing CP regularizers]\label{exp:existing_CP_sparse_regularizers}
Assigned with proper building blocks, $\Psi_{\mathrm{p}}$ reproduces $\Psi_{\mathrm{BI}}$ \cite{nikolova1998}, $\Psi_{\mathrm{GMC}}$ \cite{selesnick2017}, $\Psi_{\mathrm{LiGME}}$ \cite{abe2020} and $\Psi_{\mathrm{SSR}}$ \cite{alshabili2021}, as summarized in Table \ref{table:pSDC_reproduce_prior_arts}.
\end{example}

Moreover, since $(\psi_2\Box \phi_{\Pcal})$ is a smooth approximation of $\psi_2$, $\Psi_{\mathrm{pSDC}}$ can be regarded as a partially smoothed approximation of the function $\psi_1(\xv)-\psi_2(\Mm \xv)$. Thus for every DC-type nonconvex regularizer, we can build a pSDC regularizer as its partially smoothed approximation. We note that the class of DC functions is a very broad class such that every continuous function can be approximated to arbitrary precision by a DC function \cite[Prop. 2.3 (ii)]{horst1999}. In particular, it has been known that many nonconvex sparse regularizers are in this class \cite{le2015}. Accordingly, the pSDC regularizer certainly encompasses a large number of promising new regularizers. 

 In the following, we present an illustrative example to show how one can construct the pSDC version of a DC-type nonconvex sparse regularizer. Although our example is for sparse regularization, we note that the same technique can be applied for constructing general purpose CP regularizers (if given a DC-type regularizer as prototype).

\begin{figure*}
\centering
\begin{subfigure}{0.49\textwidth}
\centering
\includegraphics[width=\linewidth]{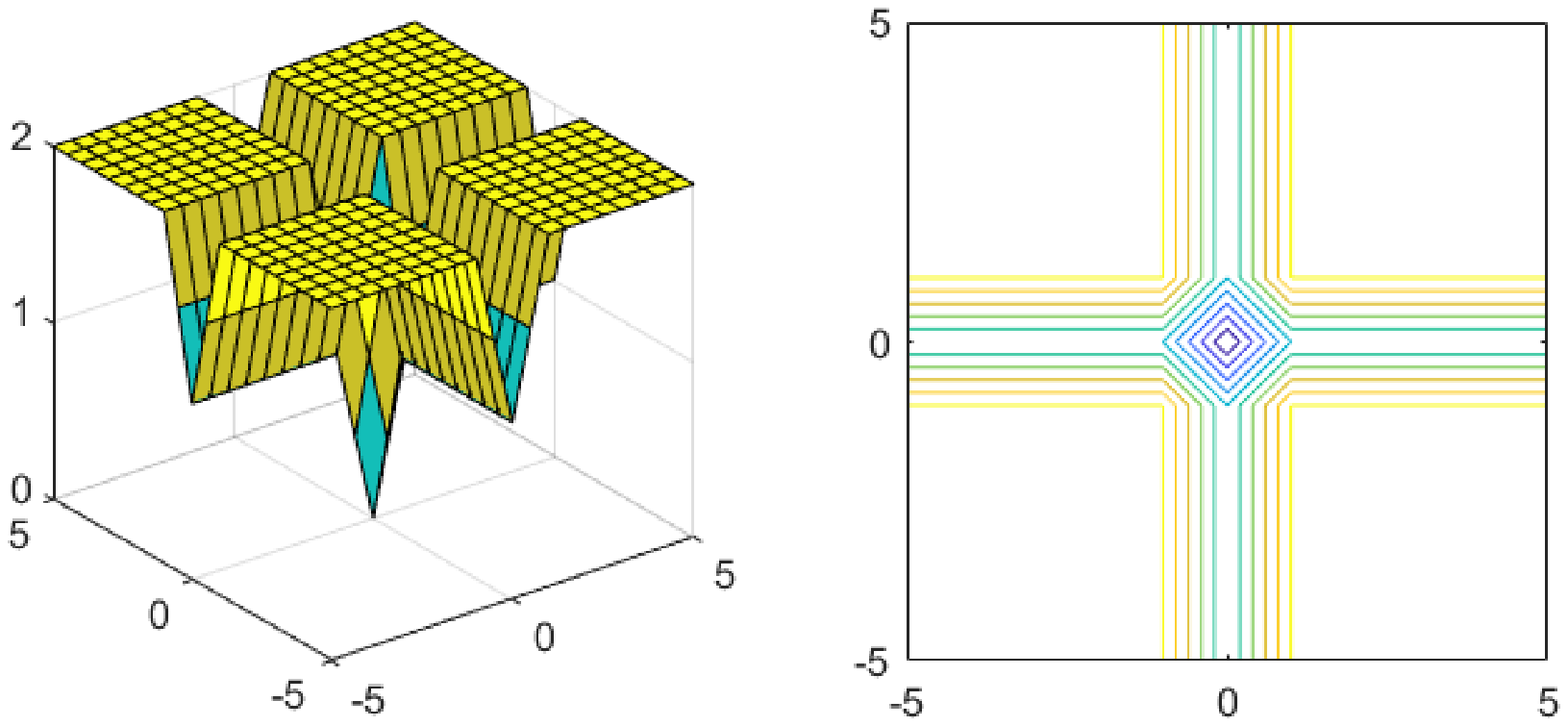}
\caption{$\Psi_{\mathrm{capped}}$}
\end{subfigure}
\begin{subfigure}{0.49\textwidth}
\centering
\includegraphics[width=\linewidth]{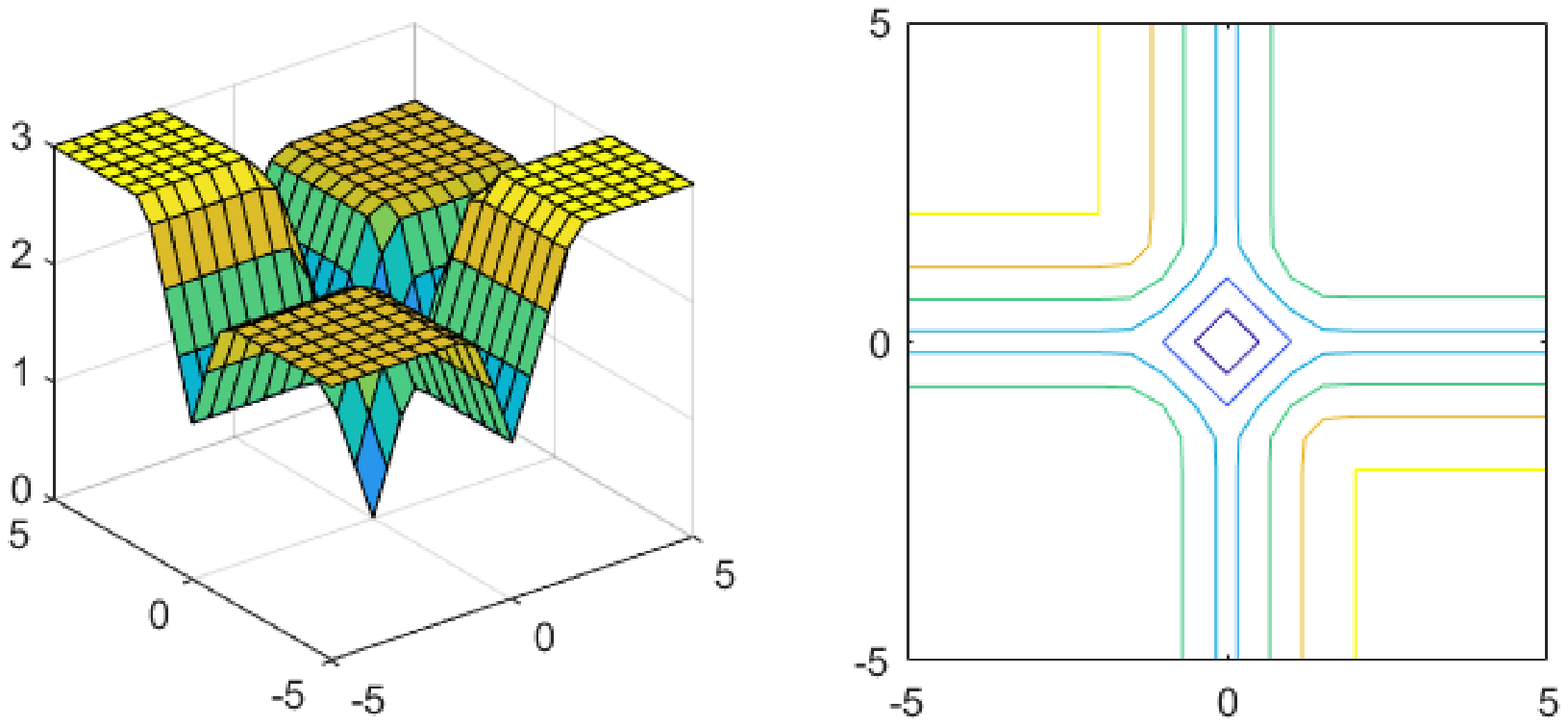}
\caption{the pSDC version of $\Psi_{\mathrm{capped}}$}
\end{subfigure}

\caption{2D surface and contours of $\Psi_{\mathrm{capped}}$ and its pSDC version.}
\label{fig:capped_l1}
\end{figure*}

\begin{example}[The pSDC version of a DC type regularizer]\label{exp:capped_l1}
We consider the capped-$l_1$ penalty \cite{peleg2008}, which is a DC-type sparse regularizer formulated as follows:
\begin{equation*}
\Psi_{\mathrm{capped}}(\xv)=\norm{\xv}_1-\sum_{i=1}^n\max\braces{|x_i|-1,0}.
\end{equation*}
To construct the pSDC version of $\Psi_{\mathrm{capped}}$, we adopt $\Mm\coloneqq \Imat_n$ and the following setting in (\ref{eq:cp_regularizer}):
\begin{align*}
\psi_1(\xv) &\coloneqq \norm{\xv}_1, \\
\psi_2(\xv) &\coloneqq \sum_{i=1}^n\max\braces{|x_i|-1,0},\\
\phi_{\Pcal}(\xv) &\coloneqq \frac{1}{2}\norm{\Bm\xv}^2_2 \;\mathrm{with}\; \Pcal\coloneqq \Bm.
\end{align*}
The resulting pSDC regularizer is depicted in Fig. \ref{fig:capped_l1} with 
\[\Bm^T\Bm=\begin{bmatrix}2 & 1\\ 1 & 2\end{bmatrix}.\] From the figure, one can verify that the pSDC version of $\Psi_{\mathrm{capped}}$ is a deformable approximation of the original regularizer, and has similar sparseness-promoting property as $\Psi_{\mathrm{capped}}$. One can imagine that the shape of the pSDC version of $\Psi_{\mathrm{capped}}$ is close to the original regularizer $\Psi_{\mathrm{capped}}$ if $\Bm^T\Bm$ is close to $\Imat_n$. However, such a matrix $\Bm^T\Bm$ usually does not satisfy the overall-convexity condition specified in Section \ref{sec:overall_convexity_conditions}. In practice, one has to trade off the sparseness-promoting performance against the overall-convexity.
\end{example}

\section{A General Nonconvexly Regularized Convex Model and Its Overall-Convexity Conditions}
\label{sec:overall_convexity_conditions}

In this section, we consider a general nonconvexly regularized convex (NRC) model which incorporates multiple pSDC regularizers along with general convex data fidelity terms and convex constraints, and we conduct a thorough discussion on its overall-convexity conditions.
\subsection{The Proposed NRC Model}\label{subsec:NRC_model}
As shown in Section \ref{sec:related_works}, the overall-convexity of the NRC model relies on the presence of "partially strongly convex" terms. Hence to derive the overall-convexity condition, we must take the whole cost function into account. Without loss of generality, we can assume that the cost function of the interested signal estimation problem is composed of:
\begin{enumerate}
\item a convex data fidelity term  $F_0\in\Gamma_0(\Real^n)$,
\item a number of convex constraints, and we assume that their intersection is a nonempty closed convex feasible set (denoted by $C_0\subset\Real^n$),

\item a number of pSDC regularizers defined as
\begin{equation*}
\Psi_{p}^{(i)}(\xv;\Pcal_i)\coloneqq \psi_{1}^{(i)}(\xv)-\parentheses{\psi_{2}^{(i)}\Box\phi_{\Pcal_i}^{(i)}}(\Mm_i\xv),
\end{equation*}
where $i=1,\dots,r$, $\Mm_i\in\Real^{q_i\times n}$, $\psi_{1}^{(i)}\in\Gamma_0(\Real^n)$, $\psi_{2}^{(i)}\in\Gamma_0(\Real^{q_i})$, $\phi^{(i)}_{\Pcal_i}\in\Gamma_0(\Real^{q_i})$
with $\Pcal_{i}$ being the shape controlling tuning parameter. In addition, we assume that the weight parameter $\lambda_{i}$ for $\Psi_{p}^{(i)}$ is positive.
\end{enumerate}
Especially, we note that by setting $\psi_{2}^{(i)}(\cdot)=\phi^{(i)}_{\Pcal_i}(\cdot)\coloneqq 0$, we have $\Psi_{p}^{(i)}(\xv;\Pcal_i)\coloneqq \psi_1^{(i)}(\xv)\in\Gamma_0(\Real^n)$, hence every proper, lower semicontinuous convex regularizer can be regarded as a special instance of pSDC regularizer.

Taking all the terms above into account yields the following regularization problem:
\begin{equation}\label{eq:example_of_NRC_model}
\underset{\xv\in C_0}{\text{minimize}}\;\;F_0(\xv)+\sum_{i=1}^r\lambda_{i}\Psi_{p}^{(i)}(\xv;\Pcal_i).
\end{equation}
Substituting the expression of $\Psi_{p}^{(i)}(\xv;\Pcal_i)$ into (\ref{eq:example_of_NRC_model}) yields
\begin{align*}
\underset{\xv\in C_0}{\text{minimize}}\;\;F_0(\xv)+\sum_{i=1}^r\lambda_{i} \psi_{1}^{(i)}(\xv)-\sum_{i=1}^r\lambda_{i} \parentheses{\psi_{2}^{(i)}\Box \phi^{(i)}_{\Pcal,i}}(\Mm_i\xv).
\end{align*}
One can verify that if we define $\zv_i\in\Real^{q_i}$,
\[\zv\coloneqq \begin{bmatrix}\zv_1^T & \cdots & \zv_r^T\end{bmatrix}^T\in\Real^{(q_1+\dots+q_r)}=:\Real^q\]
and the following:
\begin{align}
F_1(\xv) &\coloneqq F_0(\xv)+\sum_{i=1}^r\lambda_{i} \psi_{1}^{(i)}(\xv)+\iota_{C_0}(\xv),\label{eq:unification_begin}\\
F_2(\zv) &\coloneqq \sum_{i=1}^r\lambda_{i} \psi_{2}^{(i)}(\zv_i),\\
\Phi_{\mathcal{P}}(\zv) &\coloneqq \sum_{i=1}^r\lambda_{i} \phi^{(i)}_{\Pcal_i}(\zv_i)\;\;\textrm{with}\; \Pcal\coloneqq (\Pcal_1,\dots,\Pcal_r),\\
\bm{\Xi} &\coloneqq \begin{bmatrix}
\Mm_1^T, \Mm_2^T,\dots, \Mm_r^T
\end{bmatrix}^T,\label{eq:unification_end}
\end{align}
then (\ref{eq:example_of_NRC_model}) can be transformed into a more compact form:
\begin{equation}\label{eq:NRC_model}
\underset{\xv\in\Real^n}{\mathrm{minimize}}\;\; F_1(\xv)-(F_2\Box\Phi_{\Pcal})(\Xim\xv).
\end{equation}

For simplicity of discussion and superior representability, hereafter we consider the abstract cost function of (\ref{eq:NRC_model}):
\begin{equation}\label{eq:NRC_model_cost_function}
J_{p}(\xv;\mathcal{P})\coloneqq F_1(\xv)-(F_2\Box\Phi_{\mathcal{P}})(\bm{\Xi} \xv),
\end{equation}
where we assume $F_1\in\Gamma_0(\Real^n),F_2\in\Gamma_0(\Real^q)$, $\Xim\in\Real^{q\times n}$, $\Phi_{\mathcal{P}}\in\Gamma_0(\Real^q)$ is differentiable with $\nabla \Phi_{\Pcal}$ being Lipschitz continuous on $\Real^q$ with constant $1/\sigma_{\Pcal}$. Comparing (\ref{eq:NRC_model_cost_function}) with (\ref{eq:cp_regularizer}), one can verify that $J_{p}(\cdot;\Pcal)$ has a similar structure to the pSDC regularizer, and is also a DC function.

\subsection{The Overall-Convexity Conditions}
In this section, we formally establish the overall-convexity conditions with respect to $J_{p}(\cdot;\Pcal)$. In a specific regularization problem, $F_0$ and $C_0$ in (\ref{eq:example_of_NRC_model}) are usually given in advance. In addition, to ensure certain regularizing properties of the pSDC regularizer, $\psi_{1}^{(i)}$ and $\psi_{2}^{(i)}$ in $\Psi_{p}^{(i)}$ are assigned in advance, hence $F_1$ and $F_2$ in (\ref{eq:NRC_model_cost_function}) are fixed. In this case, for a fixed set of the weight parameters $\braces{\lambda_i}_{i=1}^r$, one need to find a proper value of $\Pcal$, say $\Pcal_o$ (where ``o" means overall-convexity), so that $J_{p}(\cdot;\Pcal_o)$ is convex. In the sequel, we discuss about the conditions that $
\Pcal_o$ must satisfy.

Before presenting the results, we make the following assumption to ensure that in (\ref{eq:NRC_model_cost_function}), $(F_2\Box\Phi_{\Pcal})$ is an inf-conv smooth approximation (cf. Definition \ref{def:inf_conv_approximation}) of $F_2$.

\begin{assumption}\label{assump:finite_inf_conv}
For every $\sv\in\Real^q$, $(F_2\Box\Phi_{\Pcal_o})(\sv)$ is finite, i.e., $(F_2\Box\Phi_{\Pcal_o})(\cdot)$ is real-valued.
\end{assumption}

We note that if $F_2$ and $\Phi_{\Pcal_o}$ are both bounded from below, then Assumption \ref{assump:finite_inf_conv} is satisfied, hence the assumption is not limiting. Under Assumption \ref{assump:finite_inf_conv}, we present the following overall-convexity condition, from which one can develop a number of more easily verifiable conditions.
\begin{theorem}\label{thm:general_convexity_condition}
Suppose that Assumption \ref{assump:finite_inf_conv} holds, and that $\mathcal{P}_o$ satisfies that for every $\zv\in\mathrm{dom}\;{F_2}\subset\Real^q$,
\begin{equation}\label{eq:general_convexity_condition}
M_{\zv}(\cdot;\mathcal{P}_o)\coloneqq F_1(\cdot)-\Phi_{\mathcal{P}_o}(\bm{\Xi}\cdot-\zv)
\end{equation}
is convex, then $J_{p}(\cdot;\mathcal{P}_o)\in\Gamma_0(\Real^n)$.
\end{theorem}
\begin{proof}
See Appendix \ref{app:general_convexity_condition}.
\end{proof}

Theorem \ref{thm:general_convexity_condition} indicates that the convexity of $J_{p}(\cdot;\Pcal_{o})$ in (\ref{eq:NRC_model_cost_function}) only depends on the relation between the minuend kernel function $F_1$ and the smoothing function $\Phi_{\Pcal_o}$, and is independent of the subtrahend kernel function $F_2$. More precisely, if $\Phi_{\Pcal_o}$ is ``flat" enough so that for every displacement $\zv\in\Real^q$, the convexity of $\Phi_{\Pcal_o}(\Xim\cdot-\zv)$ is overpowered by that of $F_1$, then the smoothed subtrahend function $(F_2\Box\Phi_{\Pcal_o})(\Xim\cdot)$ in (\ref{eq:NRC_model_cost_function}) is "flat" enough and is overpowered by $F_1$, which in turn guarantees the convexity of $J_{p}(\cdot;\Pcal_o)$.

Especially, we note that the overall-convexity condition described in Theorem \ref{thm:general_convexity_condition} is a sufficient condition but not a necessary one. See the following 1d counter example.
\begin{example}
Let $F_1(x)\coloneqq x^2/2$, $F_2(z)\coloneqq z^2$, $\Xi\coloneqq 1$, $\Phi_{\Pcal_o}(z)\coloneqq z^2$, then one can verify that for $z=0$, 
\begin{equation*}
M_0(x;\Pcal_o)=\frac{x^2}{2}-(x-0)^2=-\frac{x^2}{2}
\end{equation*}
is nonconvex. However, substituting the expressions of $F_1,F_2,\Xi,\Phi_{\Pcal_o}$ into (\ref{eq:NRC_model_cost_function}) yields
\begin{align*}
J_{p}(x;\Pcal_o) &=\frac{x^2}{2}-\min_{z\in\Real}\parentheses{z^2+(x-z)^2}\\
&=\frac{x^2}{2}-\frac{x^2}{2}=0,
\end{align*}
which implies that $J_{p}(\cdot;\Pcal_o)$ is convex although the convexity of $M_{z}(\cdot;\mathcal{P}_o)$ does not hold for $z=0\in\dom F_2$.
\end{example}

Theorem \ref{thm:general_convexity_condition} naturally implies the following overall-convexity condition, which is easier to verify compared to the condition proposed in Theorem \ref{thm:general_convexity_condition}.

\begin{corollary}\label{corol:convexity_condition}
1) Suppose that in (\ref{eq:NRC_model_cost_function}), $\Phi_{\Pcal_o}$ is twice continuously differentiable on $\Real^q$ and $F_1$ can be expressed as
\begin{equation}\label{eq:decomposition_of_F1}
F_1(\xv)=F_{1}^s(\xv)+F_{1}^n(\xv)
\end{equation}
with $F_{1}^s,F_{1}^n\in\Gamma_0(\Real^n)$. In addition, suppose that there exists an open set $C_1^s\supset \dom{F_1}$ such that $F_{1}^s$ is twice continuously differentiable on $C_1^s$. In this case, if $\Pcal_o$ satisfies that
\begin{equation*}
(\forall \xv\in C_1^s, \forall \zv\in\Real^q)\;\;\nabla^2F_{1}^s(\xv)\succeq \Xim^T \nabla^2 \Phi_{\Pcal_o}(\zv)\Xim,
\end{equation*}
then $J_{p}(\cdot;\Pcal_o)\in\Gamma_0(\Real^n)$.

\noindent
2) Moreover, suppose that $F_1^s$ is expressed as
\begin{align*}
F_1^s(\xv)\coloneqq \sum_{i=1}^m \xi_i(\av_i^T\xv),
\end{align*}
where $\Am\coloneqq\begin{bmatrix}\av_1,\av_2,\cdots,\av_m\end{bmatrix}^T\in\Real^{m\times n}$, $\xi_i\in\Gamma_0(\Real)$ is twice continuously differentiable on $\dom{\xi_i}$ satisfying:
\begin{equation*}
(\forall \xv\in C_1^s)\;\; \xi_i''(\av_i^T\xv)\geq \gamma_i\geq 0,\; i=1,2,\dots,m;
\end{equation*}
In this case, let $\Phi_{\Pcal}(\zv)\coloneqq \sum_{j=1}^p\eta_j(\bv_j^T\zv)$, where $\Pcal\coloneqq \Bm\coloneqq \begin{bmatrix}\bv_1,\bv_2,\cdots,\bv_p\end{bmatrix}^T\in\Real^{p\times q}$, $\eta_j\in\Gamma_0(\Real)$ is twice continuously differentiable on $\dom{\eta_j}$ satisfying
\begin{equation*}
(\forall z\in \Real)\;\; 0\leq\eta''_j(z)\leq \kappa_j,\; j=1,2,\dots,p.
\end{equation*}
Then if $\Pcal_o\coloneqq \Bm_o\in\Real^{p\times q}$ satisfies the following inequality:
\begin{equation*}
\Am^T\mathrm{diag}\parentheses{\gamma_1,\dots,\gamma_m}\Am \succeq \Xim^T\Bm_o^T\mathrm{diag}\parentheses{\kappa_1,\dots,\kappa_p}\Bm_o\Xim,
\end{equation*}
we have $J_{p}(\cdot;\Pcal_o)\in\Gamma_0(\Real^n)$.
\end{corollary} 
\begin{proof}
See Appendix \ref{app:specific_convexity_condition}.
\end{proof}

We note that Corollary \ref{corol:convexity_condition} does not only embrace the overall-convexity conditions in prior arts \cite{selesnick2017,abe2020,alshabili2021}, but also extends the applicability of existing CP regularizers. In particular, as will be shown in Section \ref{sec:poisson_denoising}, Corollary \ref{corol:convexity_condition} permits the pSDC regularizer to be applied in more general observation systems such as measurement under Poisson noise.

\section{A DC-Type Solution Algorithm}
\label{sec:dc_algorithm}
Using the overall-convexity conditions provided in Section \ref{sec:overall_convexity_conditions}, we can find a proper tuning parameter $\Pcal_o$ such that the cost function $J_{p}(\cdot;\Pcal_o)$ in (\ref{eq:NRC_model_cost_function}) is convex. However, despite convexity, the minimization of the possibly nonsmooth cost function $J_{p}(\cdot;\Pcal_o)$ is not easy. This is because most algorithms for solving (convex or nonconvex) nonsmooth composite minimization problems (e.g., the forward-backward splitting algorithm \cite{attouch2013}) require the cost function to be the sum of simple functions whose gradients or proximity operators can be computed in closed-form, whereas the presence of the involved term $-(F_2\Box\Phi_{\Pcal_o})(\Xim \cdot)$ hinders the direct application of such algorithms to the proposed NRC model (\ref{eq:NRC_model}).

To solve special classes of the NRC model (\ref{eq:NRC_model}), previous works \cite{selesnick2017,abe2020,alshabili2021,yata2021} mostly transform (\ref{eq:example_of_NRC_model}) into monotone inclusion problems \cite{bauschke2017} and develop various proximal splitting \cite{combettes2011} type algorithms for solving it. However, these algorithms  usually make stringent assumptions on the component functions in (\ref{eq:example_of_NRC_model}) (see Sec. \ref{sec:comparison_with_proximal_methods} for details). To develop a unified solution algorithm which is applicable to a more general subclass of (\ref{eq:NRC_model}), in this paper, we exploit the DC structure of $J_{p}(\cdot;\Pcal_o)$ and propose a novel solution algorithm\footnote{The idea of applying DC programming to NRC models was previously reported in \cite{zhang2021,zhang2022}. But in this paper, the considered NRC model (\ref{eq:NRC_model}) is more general and the convergence analysis is improved (see Theorem \ref{thm:convergence_of_dc_algorithm}).} for minimizing it based on DC programming \cite{le2018}.

\subsection{Selected Elements of DC Programming}\label{subsec:selected_elements_of_DC_programming}
The goal of DC programming \cite{le2018} is to find a global minimizer of a (possibly nonsmooth nonconvex) DC function. Consider a DC program of the following form:
\begin{equation}\label{eq:DC_program}
\underset{\xv\in\Real^n}{\mathrm{minimize}}\; J(\xv)\coloneqq g(\xv)-h(\xv),
\end{equation}
where $g,h\in\Gamma_0(\Real)$ with $\dom{h}=\Real^n$. A standard approach for solving (\ref{eq:DC_program}) is the simplified DCA \cite{pham1997}.

\begin{algorithm}
\textbf{Initialization:} $k=0, \xv_0\in\Real^n$.\\
Repeat the following steps.\\
\textbf{Step 1:}\quad obtain $\uv_k\in\partial h\parentheses{\xv_k}$.

\textbf{Step 2:}\quad  compute $\xv_{k+1}$ by
\[\xv_{k+1}\in\underset{\xv\in\Real^n}{\argmin}\; g(\xv)-\innerp{\uv_{k}}{\xv},\]
\quad\quad\quad$\,\,$ and update $k\leftarrow k+1$.

\caption{The simplified DCA for solving (\ref{eq:DC_program})}\label{alg:simplified_dca}
\end{algorithm}

As summarized in Algorithm \ref{alg:simplified_dca}, in every iteration, the simplified DCA finds a subgradient $\uv_k\in\partial h(\xv_k)$. According to the definition of $\partial h(\xv_k)$ (cf. Section \ref{subsec:notation}), $\uv_k$ satisfies that
\begin{equation}\label{eq:h_is_minorant}
(\forall \xv\in\Real^n)\;\; h(\xv)\geq h_k(\xv)\coloneqq h(\xv_k)+\innerp{\uv_k}{\xv-\xv_k},
\end{equation}
hence the affine function $h_k$ is an affine minorant of $h$ which coincides with $h$ at $\xv_k$.
Substituting $h_k$ for $h$ in (\ref{eq:DC_program}), we define the convex surrogate cost function $J_k(\xv)\coloneqq g(\xv)-h_k(\xv)$. According to (\ref{eq:h_is_minorant}), the following holds:
\begin{equation*}
(\forall\xv\in\Real^n)\;\; J_k(\xv)\geq J(\xv),
\end{equation*}
thus $J_k(\xv)$ is a majorant of $J(\xv)$. The new estimate $\xv_{k+1}$ is selected as a global minimizer of this majorant\footnote{One may notice that the underlying idea of the simplified DCA is similar to that of the majorization minimization (MM) algorithm \cite{sun2017}. However, we remark that the simplified DCA is not an implementation of MM since it does not necessarily satisfy (A2.2) in \cite{sun2017}.}:
\begin{align*}
&\xv_{k+1} \in \underset{\xv\in\Real^n}{\argmin}\; J_k(\xv)= g(\xv)-h(\xv_k)-\innerp{\uv_{k}}{\xv-\xv_k}\\
&\iff\xv_{k+1}\in \underset{\xv\in\Real^n}{\argmin}\; g(\xv)-\innerp{\uv_{k}}{\xv}.
\end{align*}

Since $J$ in (\ref{eq:DC_program}) is possibly nonconvex and nonsmooth, the convergence of Algorithm \ref{alg:simplified_dca} to a global minimizer of (\ref{eq:DC_program}) is not guaranteed. Instead, it has been proved in \cite{pham1997} that the sequence generated by Algorithm \ref{alg:simplified_dca} subsequentially converges to a "critical point" of $g-h(=:J)$, i.e., a point $\xv_{*}$ satisfying
\[\partial g(\xv_{*})\cap\partial h(\xv_*)\neq \emptyset.\]
However, in general, being a critical point of $g-h$ is a necessary but insufficient condition for being a local minimizer of $J$. Especially, even if the cost function $J$ is overall-convex, a critical point of $g-h$ may not be the global minimizer of $J$; see the following 1d counter example.

\begin{example}\label{exp:critical_is_not_global_minimizer}
We define $g$ and $h$ as follows:
\begin{equation*}
g(x)\coloneqq \begin{cases}
|x|, & \textrm{if }-1\leq x\leq 1,\\
2|x|-1, & \textrm{otherwise}.
\end{cases}
\end{equation*}
\begin{equation*}
h(x)\coloneqq \begin{cases}
0, &\textrm{if }-1\leq x\leq 1,\\
|x|-1, & \textrm{otherwise}.
\end{cases}
\end{equation*}
Then one can verify the following:
\begin{enumerate}
\item $g,h\in\Gamma_0(\Real)$.
\item $\partial g(1)=[1,2]$, $\partial h(1)=[0,1]$.
\item $J(x)\coloneqq g(x)-h(x)=|x|$ is convex, and it has a unique global minimizer at $x=0$.
\end{enumerate}
From 2), we have $\partial g(1)\cap\partial h(1)=\braces{1}\neq \emptyset$, which implies that $x_*=1$ is a critical point of $g-h$. However, $x_*=1\neq 0$ is not the global minimizer of $J$.
\end{example}

\subsection{Derivation of the Algorithm}
Since $J_{p}(\cdot;\Pcal_o)$ is a DC function, (\ref{eq:NRC_model}) can be regarded as a special instance of (\ref{eq:DC_program}) with overall-convexity, which admits the following decomposition:
\begin{align}
g(\xv) &\coloneqq F_1(\xv),\nonumber\\
h(\xv) &\coloneqq \bar{F}_2(\xv;\Pcal_o)\coloneqq (F_2\Box \Phi_{\Pcal_o})(\Xim \xv).\label{eq:F2_bar}
\end{align}

Nonetheless, the simplified DCA cannot be applied directly to $J_{p}(\cdot;\Pcal_o)$. This is because $h$ does not have a closed form expression for its function value, not to mention its subdifferential, which poses a difficulty in computing $\uv_k\in\partial h(\xv_k)$. To overcome this obstacle, we show in Lemma \ref{lemma:nabla_bar_F2} that $\uv_k$ can be computed through solving a common convex program. Before presenting the result, we make the following assumption (note that Assumption \ref{assump:exact_inf_conv} implies Assumption \ref{assump:finite_inf_conv}).

\begin{assumption}\label{assump:exact_inf_conv}
The infimal convolution $F_2\Box\Phi_{\Pcal_o}$ is exact, i.e., for every $\sv\in\Real^q$, there exists $\zv_{\sv}\in\Real^q$ such that
\begin{equation}
F_2(\zv_{\sv})+\Phi_{\Pcal_o}(\sv-\zv_{\sv})=\inf_{\zv\in\Real^q}\braces{F_2(\zv)+\Phi_{\Pcal_o}(\sv-\zv)}.
\end{equation}
\end{assumption}

\begin{lemma}\label{lemma:nabla_bar_F2}
Suppose that Assumption \ref{assump:exact_inf_conv} holds. Then:
\begin{enumerate}
\item $\bar{F}_2(\cdot;\Pcal_o)$ in (\ref{eq:F2_bar}) is in $\Gamma_0(\Real^n)$ and is continuously differentiable.
\item let $\zv_{\xv}$ be the point satisfying:
\begin{equation}\label{eq:convex_program_in_nabla_F2}
\zv_{\xv}\in \underset{{\zv\in\Real^q}}{\argmin}\; F_2(\zv)+\Phi_{\Pcal_o}(\Xim\xv-\zv),
\end{equation}
then $\nabla \bar{F}_2(\xv;\Pcal_o)=\Xim^T\nabla \Phi_{\Pcal_o}(\Xim\xv-\zv_{\xv})$.
\end{enumerate}
\end{lemma}
\begin{proof}
The result 1) follows from Proposition \ref{prop:inf_conv_smooth_approximation}-1). Applying the chain rule to Proposition \ref{prop:inf_conv_smooth_approximation}-2) yields the result 2).
\end{proof}

Lemma \ref{lemma:nabla_bar_F2}-2) indicates that we can obtain $\uv_k\in\partial h(\xv_k)\coloneqq \braces{\nabla \bar{F}_2(\xv_k;\Pcal_o)}$ via solving the convex program (\ref{eq:convex_program_in_nabla_F2}). Applying this idea to Algorithm \ref{alg:simplified_dca} yields\footnote{We should note that using Lemma \ref{lemma:nabla_bar_F2}, it is also possible to apply the nonconvex forward-backward splitting algorithm to (\ref{eq:NRC_model}) by setting $-(F_2\Box\Phi_{\Pcal_o})(\Xim\cdot)$ as $h$ and setting $F_1$ as $g$ (cf. \cite[Eq. (50)]{attouch2013}). However, one can verify that the resultant algorithm is essentially equivalent to the so-called proximal linearized DC algorithm \cite{souza2016}, which can be regarded as an alternative for the proposed Algorithm \ref{alg:proposed_dca} but does not make an essential difference.} the proposed Algorithm \ref{alg:proposed_dca}. Therefore, we can solve the involved NRC model (\ref{eq:NRC_model}) by solving a sequence of simpler convex programs (\ref{eq:alg1_step1}) and (\ref{eq:alg1_step2}).

\begin{algorithm}
\textbf{Initialization:} $k=0, \xv_0\in\Real^n.$\\
Repeat the following steps until convergence.\\
\textbf{Step 1:}\quad obtain $\zv_k$ by
\begin{equation}\label{eq:alg1_step1}
\zv_k\in\underset{\zv\in\Real^q}{\argmin}\; F_2(\zv)+\Phi_{\Pcal_o}(\Xim \xv_k-\zv),
\end{equation}
\quad\quad\quad$\,\,$ and compute $\uv_k=\Xim^T\nabla \Phi_{\Pcal_o}(\Xim \xv_k-\zv_k)$.

\textbf{Step 2:}\quad compute $\xv_{k+1}$ by
\begin{equation}\label{eq:alg1_step2}
\xv_{k+1}\in\underset{\xv\in\Real^n}{\argmin}\; F_1(\xv)-\innerp{\uv_k}{\xv},
\end{equation}
\quad\quad\quad$\,\,$ and update $k\leftarrow k+1$.

\caption{Proposed DC algorithm for solving (\ref{eq:NRC_model})}\label{alg:proposed_dca}
\end{algorithm}

\renewcommand{\arraystretch}{1.5}
\begin{table*}
\centering
\caption{Parameter settings for each algorithm.}\label{table:algorithm_parameters}
\begin{tabular}{|c|c|} 
 \hline
 Algorithm & Hyperparameters \\
 \hline\hline
 L1 (ISTA \cite[Sec. 2.1]{beck2009}) & $t_k\equiv 1.99/\norm{\Am^T\Am}_2$\\
 \hline
 GMC \cite[Prop. 15]{selesnick2017} & $\mu=1.99/\max\braces{1,\rho/(1-\rho)}\norm{\Am^T\Am}_2$\\
 \hline
 LiGME \cite[Alg. 1]{abe2020} &  $\kappa=1.01$, and $\sigma,\tau$ are defined as \cite[Footnote 7]{abe2020}\\
 \hline
 SSR \cite[Alg. 1]{alshabili2021} (Quad/Log) & $\alpha=0.99/2L$ with $L\coloneqq \sqrt{(1+4\rho^2)}\norm{\Am^T\Am}_2$ \\
 \hline
 Inner loop of Algorithm \ref{alg:proposed_dca} for (\ref{eq:alg1_step1}) (Quad/Log) &  $t_k\equiv 1.99\lambda/\rho\norm{\Am^T\Am}_2$, $\epsilon_1=10^{-3}$, $\delta_1=10^{-10}$\\
 \hline
 Inner loop of Algorithm \ref{alg:proposed_dca} for (\ref{eq:alg1_step2}) (Quad/Log) & $t_k\equiv 1.99/\norm{\Am^T\Am}_2$, $\epsilon_2=10^{-3}$, $\delta_2=10^{-10}$\\
 \hline
\end{tabular}
\end{table*}

\subsection{Convergence Properties}
As shown in Section \ref{subsec:selected_elements_of_DC_programming}, $\parentheses{\xv_k}_{k\in\Natural}$ generated by Algorithm \ref{alg:simplified_dca} is only guaranteed to converge to a critical point of $g-h$, which may not be a global minimizer even if the overall-convexity condition holds. Since Algorithm \ref{alg:proposed_dca} is essentially the same algorithm as Algorithm \ref{alg:simplified_dca}, one may concern that Algorithm \ref{alg:proposed_dca} yields worse convergence guarantee in comparison with existing proximal splitting type methods \cite{selesnick2017,abe2020,alshabili2021,yata2021}. 

Fortunately, the cost function $J_{p}(\cdot;\Pcal_o)$ to be solved by Algorithm \ref{alg:proposed_dca} has a favourable property in addition to the overall-convexity, that is, the subtrahend function $\bar{F}_2(\cdot;\Pcal_o)$ is differentiable. Exploiting this property, we have proved in Theorem \ref{thm:convergence_of_dc_algorithm} that every critical point of $F_1(\cdot)-\bar{F}_2(\cdot;\Pcal_0)$ is a global minimizer of $J_{p}(\cdot;\Pcal_o)$ (cf. Appendix \ref{app:convergence_of_DC_algorithm}), whereby we can prove the convergence of Algorithm \ref{alg:proposed_dca} to a global minimizer of $J_{p}(\cdot;\Pcal_o)$. Accordingly, the reliability of the proposed Algorithm \ref{alg:proposed_dca} is ensured.

Before presenting the result, since $J_{p}(\cdot;\Pcal_o)$ is supposed to be the cost function of a signal estimation problem, we assume that the solution set of (\ref{eq:NRC_model}) is nonempty and bounded.
\begin{assumption}\label{assump:boundedness_of_argmin}
The set of global minimizers of $J_{p}(\cdot;\Pcal_o)$, i.e., $\argmin_{\xv\in\Real^n}\; J_{p}(\xv;\Pcal_o)$ is nonempty and bounded.
\end{assumption}
The convergence properties of Algorithm \ref{alg:proposed_dca} is as follows.

\begin{theorem}\label{thm:convergence_of_dc_algorithm}
Suppose that Assumption \ref{assump:finite_inf_conv}, \ref{assump:exact_inf_conv}, \ref{assump:boundedness_of_argmin} and the overall-convexity condition specified in Theorem \ref{thm:general_convexity_condition} hold. Let $(\xv_k)_{k\in\mathbb{N}}$ and $(\uv_k)_{k\in\Natural}$ be sequences generated by Algorithm \ref{alg:proposed_dca}, and let
\[\alpha=\min_{\xv\in\Real^n}J_{p}(\xv;\Pcal_o),\]
then the following holds:
\begin{enumerate}
\item for every $k\in\Natural$, $J_{p}(\xv_{k+1};\Pcal_o)\leq J_{p}(\xv_k;\Pcal_o)$.

\item $(\xv_k)_{k\in\Natural}$ and $(\uv_k)_{k\in\Natural}$ are bounded.

\item every limit point of $(\xv_k)_{k\in\Natural}$, denoted by $\bar{\xv}$, is a global minimizer of $J_{p}(\cdot;\Pcal_o)$, i.e.,
\begin{equation*}
J_{p}(\bar{\xv};\Pcal_o)=\alpha.
\end{equation*}

\item $\lim_{k\to +\infty}J_{p}(\xv_k;\Pcal_o)=\alpha$.
\end{enumerate}
\end{theorem}
\begin{proof}
See Appendix \ref{app:convergence_of_DC_algorithm}.
\end{proof}

One may be interested in the convergence rate of Algorithm \ref{alg:proposed_dca}. We note that if $J_{p}(\cdot;\Pcal_o)$ is a very simple cost function such that its so-called Lojasiewicz exponent $\theta$ is known, the convergence rate analysis is possible. More precisely, let $\xv_\infty$ be the limit point of $\parentheses{\xv_k}_{k\in\Natural}$, then if $\theta\in(1/2,1)$, the convergence of $\norm{\xv_k-\xv_{\infty}}_2$ is sublinear; if $\theta\in(0,1/2]$, the convergence of $\norm{\xv_k-\xv_{\infty}}_2$ is linear (see \cite[Thm. 3.3]{le2018subanalytic} for details). However, we note that in general the Lojasiewicz exponent of a given function is unknown, hence it is difficult to determine the convergence rate of Algorithm \ref{alg:proposed_dca} in practice.
\subsection{Implementation Details}
In practice, when implementing Algorithm \ref{alg:proposed_dca}, one usually needs to invoke some iterative algorithms to solve the convex subproblems (\ref{eq:alg1_step1}) and (\ref{eq:alg1_step2}), which leads to two inner loops. With respect to these inner loops, several points need to be discussed; see the following remarks.

\begin{remark}[Stopping criteria of the inner loops]\label{remark:termination_rules}
In general, it takes an infinite number of inner iterations to obtain an exact solution to (\ref{eq:alg1_step1}) or (\ref{eq:alg1_step2}), which leads to non-terminating inner loops. To resolve this impracticality, in practice we often adopt certain stopping criteria to obtain high-quality inexact solutions to the subproblems. Let $\zv_{k,i}$ be the $i$th inner loop estimate for $\zv_k$ in (\ref{eq:alg1_step1}), and let $\xv_{k+1,j}$ be the $j$th inner loop estimate for $\xv_{k+1}$ in (\ref{eq:alg1_step2}). In this paper, we recommend the following stopping criteria which require the update caused by the current inner iteration to be sufficiently small compared to the last inner estimate. More precisely, we suggest terminating the inner loop for (\ref{eq:alg1_step1}) and set $\zv_k\coloneqq \zv_{k,j+1}$ if 
\begin{equation*}
\norm{\zv_{k,j+1}-\zv_{k,j}}_2\leq \epsilon_1(\norm{\zv_{k,j}}_2+\delta_1),
\end{equation*}
and we suggest terminating the inner loop for (\ref{eq:alg1_step2}) and set $\xv_{k+1}\coloneqq \xv_{k+1,j+1}$ if 
\begin{equation*}
\norm{\xv_{k+1,j+1}-\xv_{k+1,j}}_2\leq \epsilon_2(\norm{\xv_{k+1,j}}_2+\delta_2),
\end{equation*}
where $\epsilon_1,\epsilon_2>0$, $\delta_1,\delta_2\geq 0$.
\end{remark}

Applying the stopping criteria above leads to an inexact version of Algorithm \ref{alg:proposed_dca}. We note that although the convergence of this inexact DC algorithm has been verified by extensive experiments (cf. Section \ref{sec:numerical_experiments}), it is difficult to prove it theoretically. Some recent works \cite{souza2016,deoliveira2019,lu2019} have proposed other inexact versions of DC algorithms, with respect to which the convergence to critical points can be guaranteed. Nevertheless, the stopping criteria proposed in these studies usually do not possess implementability (i.e., the stopping criteria should be easily verifiable) and achievability (i.e., the stopping criteria should be satisfied within finite inner iterations) at the same time, hence are not qualified for solving (\ref{eq:NRC_model}).

Fortunately, during the reviewing process of this paper, we have developed a novel inexact DC algorithm \cite{zhang2023} which can tackle the aforementioned difficulties for a certain class of DC functions \footnote{To be more precise, we mean the class of DC functions whose minuend part can be decomposed as the pointwise maximum of finitely many convex smooth functions.}. However, to extend this new algorithm to a more general class of the NRC model (\ref{eq:NRC_model}) requires more effort, and we would like to leave it for future research.

\begin{remark}[Initialization of the inner loops]
Every time we invoke an inner iterative algorithm for solving (\ref{eq:alg1_step1}) (or (\ref{eq:alg1_step2})), we need to reassign the initial guess $\zv_{k,0}$ (or $\xv_{k+1,0}$). In this paper, to make full use of the past iterations, we recommend adopting the warm start strategy, i.e., we set $\zv_{k,0} \coloneqq \zv_{k-1}$ and $\xv_{k+1,0}\coloneqq \xv_{k}$ for every $k\geq 0$. In this case, the initialization of the inner loops ultimately reduces to the initialization of $\zv_{-1}$ and $\xv_0$. As will be shown in Section \ref{sec:comparison_of_algorithms}, in the standard sparse recovery problem, assigning $\zv_{-1}\coloneqq \mathbf{0}_p$ and $\xv_{0}\coloneqq \mathbf{0}_n$ leads to empirically fast convergence, of which an intuitive interpretation is given in Remark \ref{remark:zero_initialization}. 
\end{remark}

\subsection{Comparison with Proximal Splitting Type Algorithms}\label{sec:comparison_with_proximal_methods}
Conventional studies \cite{selesnick2017,abe2020,alshabili2021,yata2021} which adopt proximal splitting type algorithms usually assume the following in (\ref{eq:example_of_NRC_model}):
\begin{enumerate}
\item the data fidelity $F_0$ is quadratic.
\item the kernel functions $\psi_{1}^{(i)},\psi_{2}^{(i)}$ in $\Psi_{p}^{(i)}$ are proximable functions or their compositions with linear operators.
\item the constraint $C_0$ is of split feasibility type \cite{censor2005}, i.e., $C_0$ can be rewritten as
\begin{equation*}
C_0\coloneqq \braces{\xv\in\Real^n\mid \Am_i\xv\in C_i, \textrm{ for } i=1,\dots, s}.
\end{equation*}
where for $i=1,\dots,s$, $\Am_i$ is a linear operator and $C_i$ is a ``simple" closed convex set (by ``simple", we mean the projection onto $C_i$ can be computed to high precision efficiently).
\end{enumerate}

Due to these assumptions, existing proximal splitting type algorithms \cite{selesnick2017,abe2020,alshabili2021,yata2021} are applicable to a limited subclass of the proposed NRC model (\ref{eq:NRC_model}).

In contrast, Algorithm \ref{alg:proposed_dca} is applicable to (\ref{eq:NRC_model}) as long as (\ref{eq:alg1_step1}) and (\ref{eq:alg1_step2}) can be solved by some iterative algorithms. In particular, substituting (\ref{eq:unification_begin}-\ref{eq:unification_end}) into (\ref{eq:alg1_step1}) and (\ref{eq:alg1_step2}) yields Algorithm \ref{alg:proposed_dca_special} for solving (\ref{eq:example_of_NRC_model}), and one can verify that if the assumptions made in proximal splitting type algorithms hold, then (\ref{eq:alg3_step1}) and (\ref{eq:alg3_step2}) can be solved efficiently via primal-dual splitting algorithms \cite{combettes2012,condat2013}. Moreover, various advanced optimization techniques can be easily incorporated into the implementation of Algorithm \ref{alg:proposed_dca_special} to handle difficult terms or to yield faster convergence, for example:
\begin{enumerate}
\item Nesterov acceleration \cite{nesterov1983} or Anderson acceleration \cite{walker2011} can help achieve faster convergence of the inner loops.

\item Line search techniques such as Wolfe conditions \cite{nocedal2006} can be used to handle differentiable data fidelity term $F_0$ whose Lipschitz constant of the gradient is unknown.

\item Adopting interior point methods \cite[Ch. 11]{boyd2004} as the inner algorithm for (\ref{eq:alg3_step2}), in principle we can deal with certain non-split-feasibility type constraints such as the following type:
\begin{equation*}
C_0\coloneqq \braces{\xv\in\Real^n\mid f_i(\xv)\leq c_i, i=1,\dots, c},
\end{equation*}
where $f_i$ is convex and twice continuously differentiable.
\end{enumerate}

In addition, proximal splitting type algorithms \cite{abe2020,yata2021} deal with the convex constraint and multiple regularizers by lifting the original problem to a higher dimensional space. In contrast, Algorithm \ref{alg:proposed_dca_special} does not handle constraints and multiple regularizers directly in a single monotone inclusion problem, but can deal with them in the simpler convex subproblem (\ref{eq:alg3_step1}) and (\ref{eq:alg3_step2}). Hence by employing proper inner algorithms (e.g., three-operator splitting scheme \cite{davis2017}), Algorithm \ref{alg:proposed_dca_special} does not necessarily increase the problem dimension.

\begin{figure*}
\begin{multicols}{2}
\centering
\includegraphics[height=5cm]{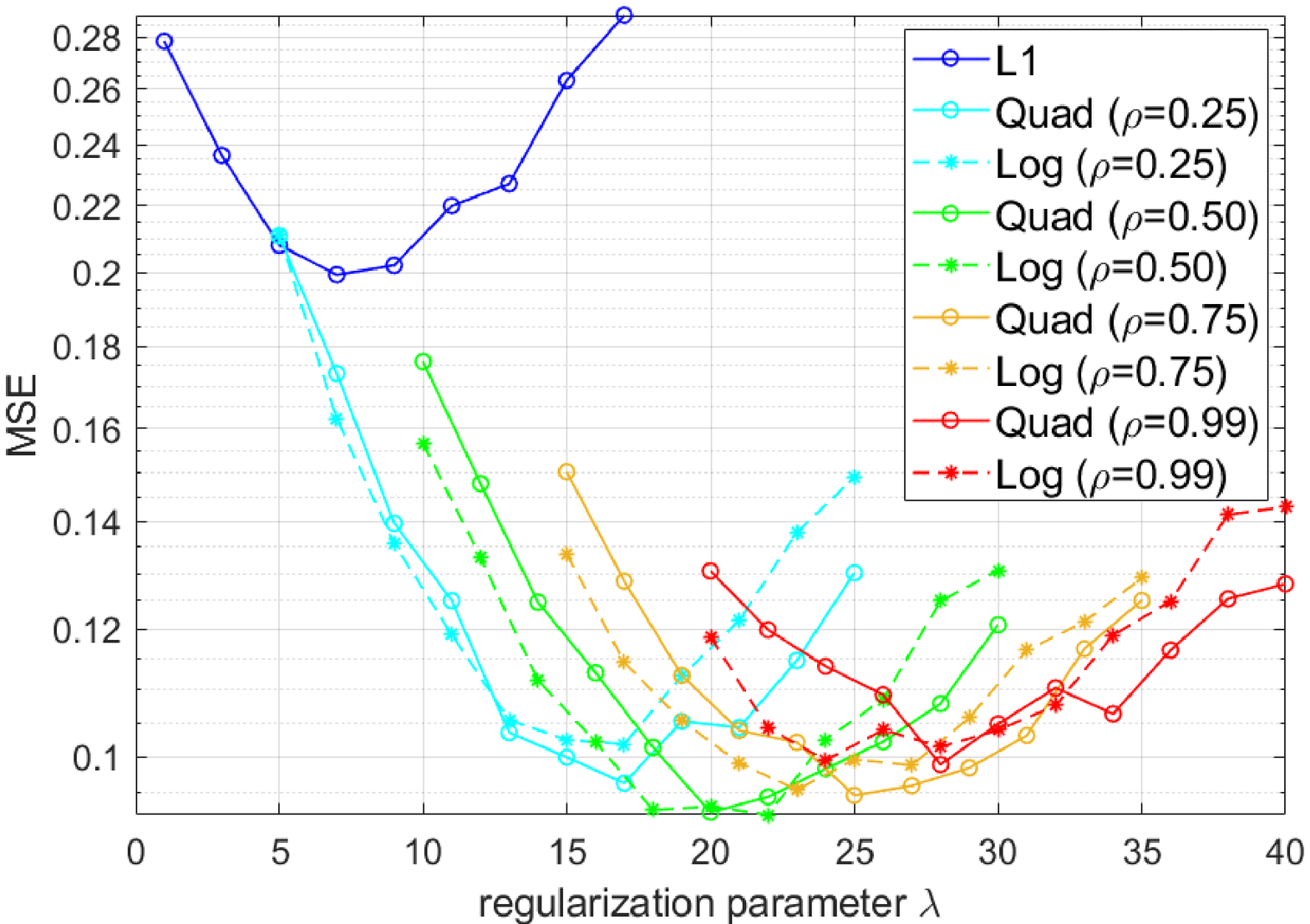}
\caption{MSE vs hyperparameters.}\label{fig:MSE_vs_lambda}
\par
\centering
\includegraphics[height=5cm]{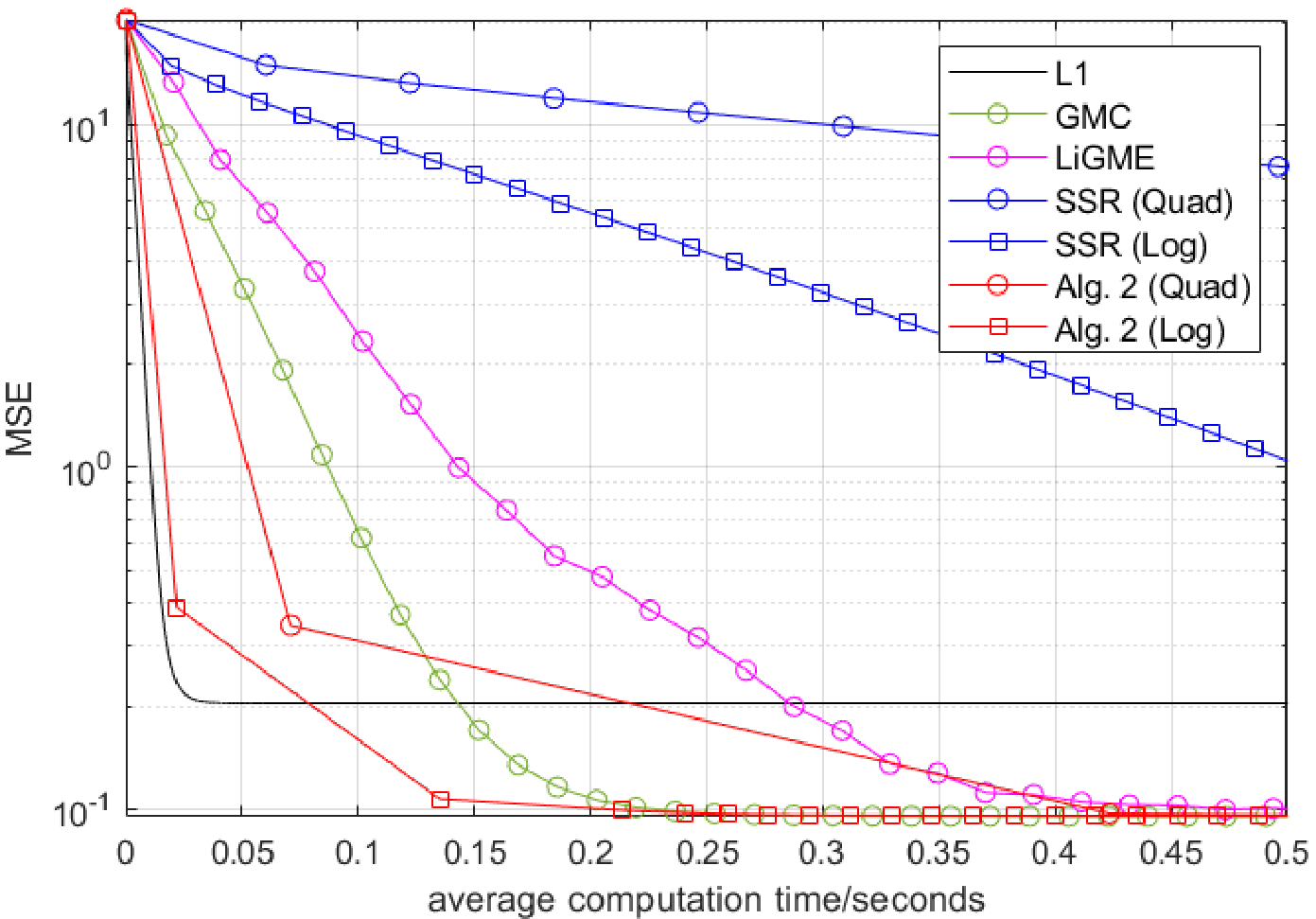}
\caption{Performance of each algorithm.}\label{fig:algorithm_comparison}
\par
\end{multicols}
\end{figure*}

\begin{figure*}[!t]
\centering
\includegraphics[width=\linewidth]{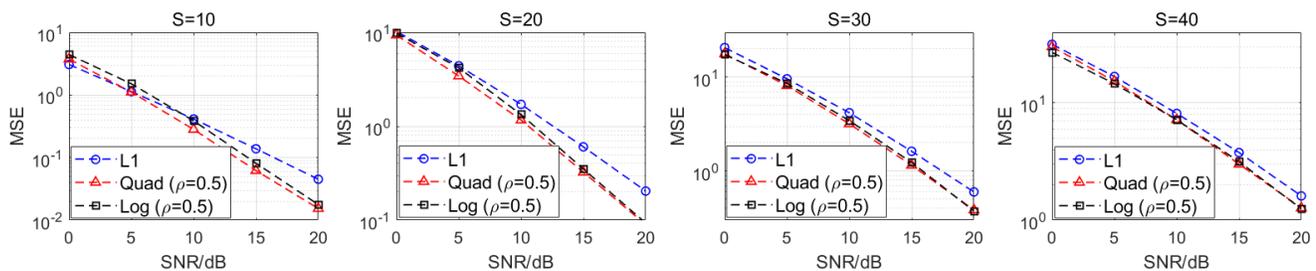}
\caption{MSE vs SNR with different sparseness level $S$.}\label{fig:MSE_vs_SNR}
\end{figure*}

\begin{algorithm}
\textbf{Initialization:} $k=0, \xv_0\in\Real^n.$\\
Repeat the following steps until convergence.\\
\textbf{Step 1:}\quad for $i=1,\dots,r$, obtain $\zv_k^{(i)}$ by
\begin{equation}\label{eq:alg3_step1}
\zv_k^{(i)}\in\underset{\zv\in\Real^{q_i}}{\argmin}\; \psi_{2}^{(i)}(\zv)+\phi^{(i)}_{\Pcal_{o,i}}(\Mm_i \xv_k-\zv),
\end{equation}
\quad\quad\quad$\,\,$ and compute
\[\uv_k=\sum_{i=1}^r \Mm_i^T\nabla \phi^{(i)}_{\Pcal_{o,i}}\parentheses{\Mm_i \xv_k-\zv_k^{(i)}}.\]

\textbf{Step 2:}\quad compute $\xv_{k+1}$ by
\begin{equation}\label{eq:alg3_step2}
\xv_{k+1}\in\underset{\xv\in C_0}{\argmin}\; F_0(\xv)+\sum_{i=1}^r\lambda_{i} \psi_{1}^{(i)}(\xv)-\innerp{\uv_k}{\xv},
\end{equation}
\quad\quad\quad$\,\,$ and update $k\leftarrow k+1$.

\caption{Proposed DC algorithm for solving (\ref{eq:example_of_NRC_model})}\label{alg:proposed_dca_special}
\end{algorithm}

\section{Numerical Experiments}
\label{sec:numerical_experiments}

This section evaluates the proposed pSDC regularizer and the proposed DC algorithm for solving (\ref{eq:NRC_model}) in numerical experiments. All experiments were performed using MATLAB (R2020b) on a computer with Intel(R) Core(TM) i7-6700T CPU 2.80, 16GB (RAM), under Windows 10, 64Bits.

\subsection{Standard Sparse Recovery}\label{sec:experiments:sparse_recovery}
We first conduct numerical experiments in a scenario of standard sparse recovery problems. Through this simplest example, we can understand many properties of CP regularizers and NRC models, and we can verify their superiority over conventional convex regularization models.

We generate the synthetic data as follows: the sparse signal to be estimated $\xv_{\star}\in\Real^{1000}$ is generated by "sprandn" function in MATLAB, where $S$ out of 1000 components are nonzero. The observed signal is $\yv=\Am\xv_\star+\bm{\epsilon}$, where the entries of $\Am\in\Real^{200\times 1000}$ follow the standard normal distribution, $\bm{\epsilon}$ is additive white Gaussian noise. The signal-to-noise ratio (SNR) is defined as $\textrm{SNR}\coloneqq 20\log_{10}\parentheses{\norm{\Am\xv_\star}_2/\norm{\bm{\epsilon}}_2} \;(\textrm{dB})$. We estimate $\xv_\star$ from $\yv$ by solving the following NRC model:
\begin{align}\label{eq:experiment:sparse_recovery}
\underset{\xv\in\Real^{1000}}{\mathrm{minimize}}\;\;  &\frac{1}{2}\norm{\yv-\Am\xv}^2_2+\lambda\Psi_{p}(\xv;\rho)\nonumber\\
\coloneqq &  \frac{1}{2}\norm{\yv-\Am\xv}^2_2+\lambda\norm{\xv}_1-\lambda(\norm{\cdot}_1\Box\phi_\rho)(\xv),
\end{align}
where the sparseness-promoting pSDC regularizer $\Psi_{p}(\cdot;\rho)$ is constructed by setting $\psi_1(\cdot)=\psi_2(\cdot)\coloneqq \norm{\cdot}_1, \Mm\coloneqq\Imat_n$ in (\ref{eq:cp_regularizer}), and we consider two types of smoothing functions $\phi_{\rho}$:
\begin{enumerate}
\item the quadratic function (Quad): $\phi_{\rho}(\cdot)\coloneqq \frac{1}{2}\norm{\Bm_\rho\cdot}^2_2$,

\item the logarithmic function (Log): $\phi_{\rho}(\cdot)\coloneqq \Phi(\Bm_\rho\cdot)$ with 
\[\Phi(\zv)\coloneqq \sum_{i=1}^{200} \lvert z_i \rvert-\log_e(\lvert z_i \rvert +1),\]
\end{enumerate}
where $\Bm_\rho=\sqrt{\rho/\lambda}\Am$ is the steering matrix. Applying the transformation introduced in Section \ref{subsec:NRC_model}, one can rewrite (\ref{eq:experiment:sparse_recovery}) in the form of (\ref{eq:NRC_model}), and one can verify by Corollary \ref{corol:convexity_condition} that if $\rho\in[0,1]$, then (\ref{eq:experiment:sparse_recovery}) is convex \footnote{For the "Log" smoothing function, one can verify that $\phi_{\rho}(\zv)=\sum_{j=1}^{200} \eta(\bv_{\rho,j}^T\zv)$, where $\eta(z)\coloneqq |z|-\log_e(|z|+1)$, $\bv_{\rho,j}$ is the $j$th row vector of $\Bm_{\rho}$. Notice that $\eta''(z)\leq 1$ for every $z\in\Real$ (see \cite[Table III]{lange1990}), Corollary \ref{corol:convexity_condition}-2) can be applied to derive the overall-convexity condition.}. 

For every fixed group of hyperparameters (i.e., the smoothing function type, the values of $\lambda$ and $\rho$) of (\ref{eq:experiment:sparse_recovery}), we repeat 500 Monte Carlo runs, and evaluate the quality of the recovered signals via the mean square error (MSE):
\begin{equation*}
\mathrm{MSE}\coloneqq \sum_{t=1}^{500}\norm{\xv_\mathrm{est}^{(t)}-\xv_\star^{(t)}}^2_2,
\end{equation*}
where $\xv_\mathrm{est}^{(t)}$ and $\xv_\star^{(t)}$ are respectively the recovered signal and the true signal in the $t$th trial.

Fig. \ref{fig:MSE_vs_lambda} shows the dependency of MSE on the hyperparameters, where we set the sparseness level $S$ of $\xv_{\star}$ to be $20$, and set SNR to be 20 $\mathrm{dB}$. Especially, we note that for both smoothing functions "Quad" and "Log", setting $\rho=0$ reproduces the conventional $l_1$-regularization model, which is depicted as the "L1" curve and is considered as a baseline for sparse recovery. From Fig. \ref{fig:MSE_vs_lambda}, one can verify the superior performance of CP regularizer over conventional $l_1$-norm regardless of the smoothing function type and the value of $\rho$. 

Moreover, Fig. \ref{fig:MSE_vs_lambda} dispels a possible concern on CP regularizers: since the overall-convexity of the NRC model (\ref{eq:experiment:sparse_recovery}) requires the concave part $\lambda(\norm{\cdot}_1\Box\phi_{\rho})(\xv)$ to be overpowered by the "partially strongly convex" term $\frac{1}{2}\norm{\yv-\Am\xv}^2_2$ (cf. Section \ref{sec:GMC} and \ref{sec:overall_convexity_conditions}), one may concern that the scale of $\lambda(\norm{\cdot}_1\Box\phi_{\rho})(\xv)$ is limited to a very weak level and cannot lead to significant performance improvement. However, Fig. \ref{fig:MSE_vs_lambda} shows that for both smoothing functions and all nonzero values of $\rho$, (\ref{eq:experiment:sparse_recovery}) consistently leads to a MSE gain of around 3dB (using the optimal regularization weight $\lambda$), which verifies the promising advantage of CP regularizers.

Fig. \ref{fig:MSE_vs_SNR} shows MSE versus SNR with different sparseness level $S$. For the NRC model (\ref{eq:experiment:sparse_recovery}), we set $\rho=0.5$ for both smoothing functions. For each regularization model, we set the regularization weight as
\begin{equation*}
\lambda\coloneqq \lambda_\star\times 10^{(20-\mathrm{SNR})/20},
\end{equation*}
where $\lambda_\star$ is the optimal regularization weight of the corresponding regularization model in Fig. \ref{fig:MSE_vs_lambda}. From Fig. 4, one can verify that the MSE gain of the NRC model (\ref{eq:experiment:sparse_recovery}) compared to the $l_1$-regularization model increases with SNR and decreases with sparseness level $S$.

%\begin{figure}
%\centering
%\includegraphics[width=\linewidth]{figs/MSE_vs_lambda_poisson.eps}
%\caption{MSE vs hyperparameters.}
%\label{fig:MSE_vs_lambda_poisson}
%\end{figure}
%
%\begin{figure}
%\centering
%\includegraphics[width=\linewidth]{figs/recovered_signals_poisson.eps}
%\caption{The recovered signals.}
%\end{figure}

\subsection{Performance Comparison of Algorithms}
\label{sec:comparison_of_algorithms}
Under the same experimental settings as Fig. \ref{fig:MSE_vs_lambda}, we compare the performance of the proposed Algorithm \ref{alg:proposed_dca} with existing proximal splitting type algorithms. When choosing the quadratic smoothing function, (\ref{eq:experiment:sparse_recovery}) amounts to the GMC model, hence can be solved by Algorithm \ref{alg:proposed_dca} and the algorithms proposed for the GMC \cite{selesnick2017}, LiGME \cite{abe2020} and SSR \cite{alshabili2021} models. On the other hand, when choosing the logarithmic smoothing function, (\ref{eq:experiment:sparse_recovery}) can be solved by Algorithm \ref{alg:proposed_dca} and the algorithm proposed for the SSR \cite{alshabili2021} model. We solve the $l_1$-regularization model by ISTA \cite{combettes2005,beck2009}. 

For the NRC model (\ref{eq:experiment:sparse_recovery}), we set $\rho=0.5$ for both smoothing functions. For every regularization model, we adopt the optimal value of $\lambda$ shown in Fig. \ref{fig:MSE_vs_lambda}. In Algorithm \ref{alg:proposed_dca}, we employ ISTA \cite{combettes2005,beck2009} as the inner iterative algorithm for solving (\ref{eq:alg1_step1}) and (\ref{eq:alg1_step2}), and we use the stopping criteria described in Remark \ref{remark:termination_rules}. All algorithms are initialized with zero vectors. The other parameter settings are summarized in Table \ref{table:algorithm_parameters}.

Fig. \ref{fig:algorithm_comparison} shows the dependency of MSE of each algorithm on the average computation time. In the figure, each marker denotes 10 iterations of proximal splitting type algorithms and 1 outer iteration of Algorithm \ref{alg:proposed_dca}. The results demonstrate that for Algorithm \ref{alg:proposed_dca}, the logarithmic smoothing function leads to faster convergence than the quadratic one, which reveals a potential merit of using different smoothing functions. Moreover, despite the generality and double-loop structure of Algorithm \ref{alg:proposed_dca}, it achieves empirically fast convergence. 

Especially, for both smoothing functions, Algorithm \ref{alg:proposed_dca} is able to obtain a satisfactory estimate with merely two iterations, which can serve as a useful stopping criterion for (the outer loop of) Algorithm \ref{alg:proposed_dca} in practice. In the following remark, we present an intuitive interpretation for this phenomenon.

\begin{remark}[Interpretation of empirically fast convergence of Algorithm \ref{alg:proposed_dca}]\label{remark:zero_initialization}
We consider the NRC model (\ref{eq:experiment:sparse_recovery}) with the quadratic smoothing function. Initializing $\zv_{-1}$ and $\xv_0$ with zero vectors and substituting $\Bm_\rho=\sqrt{\rho/\lambda}\Am$ into Algorithm \ref{alg:proposed_dca}, the subproblem (\ref{eq:alg1_step1}) of the first iteration reduces to
\begin{equation*}
\zv_{0}\in\underset{\zv\in\Real^n}{\argmin}\;\; \frac{\lambda}{\rho}\norm{\zv}_1+\frac{1}{2}\norm{\Am(\mathbf{0}_n-\zv)}^2_2.
\end{equation*}
One can verify that $\zv_0=\mathbf{0}_n$, thus $\uv_0=\mathbf{0}_n$ and
\begin{equation*}
\xv_1\in\underset{\xv\in\Real^n}{\argmin}\; \frac{1}{2}\norm{\yv-\Am\xv}^2_2+\lambda\norm{\xv}_1,
\end{equation*}
which means $\xv_1$ is a solution to the $l_1$-regularization problem. Accordingly, the first iteration of Algorithm \ref{alg:proposed_dca} produces a high-quality guess, and can be interpreted as a bit luxurious preprocessing step. Started from $\xv_1$, one can expect that $\xv_2$ can possibly be a satisfactory estimate and hence Algorithm \ref{alg:proposed_dca} would approach to the true signal rapidly.
\end{remark}

\begin{figure*}
	\centering
	\begin{minipage}{\columnwidth}
		\centering
		\includegraphics[width=\textwidth]{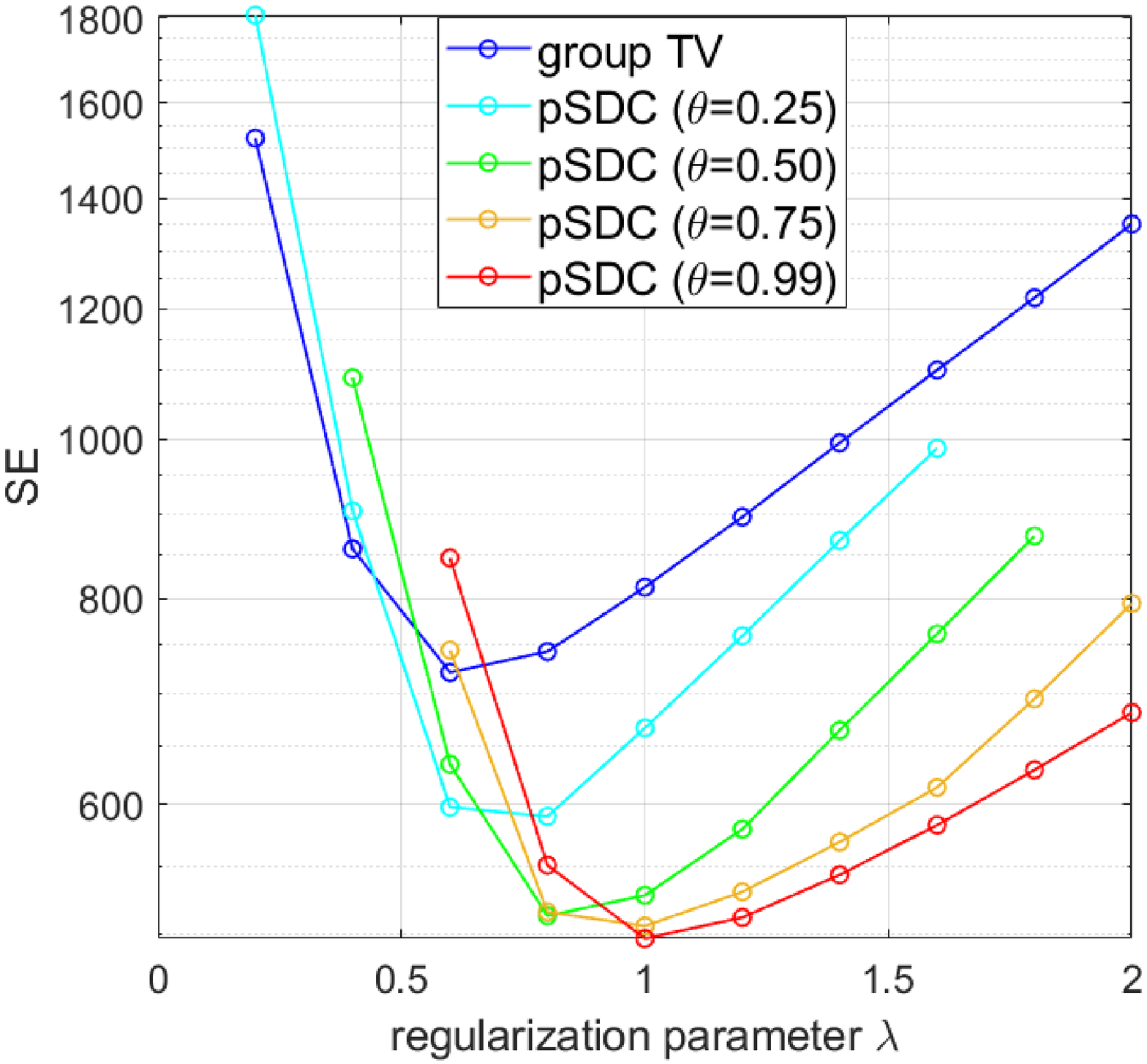}
		\caption{SE vs hyperparameters.}\label{fig:MSE_vs_lambda_poisson}
	\end{minipage}%
	\begin{minipage}{\columnwidth}
		\centering
		\includegraphics[width=\textwidth]{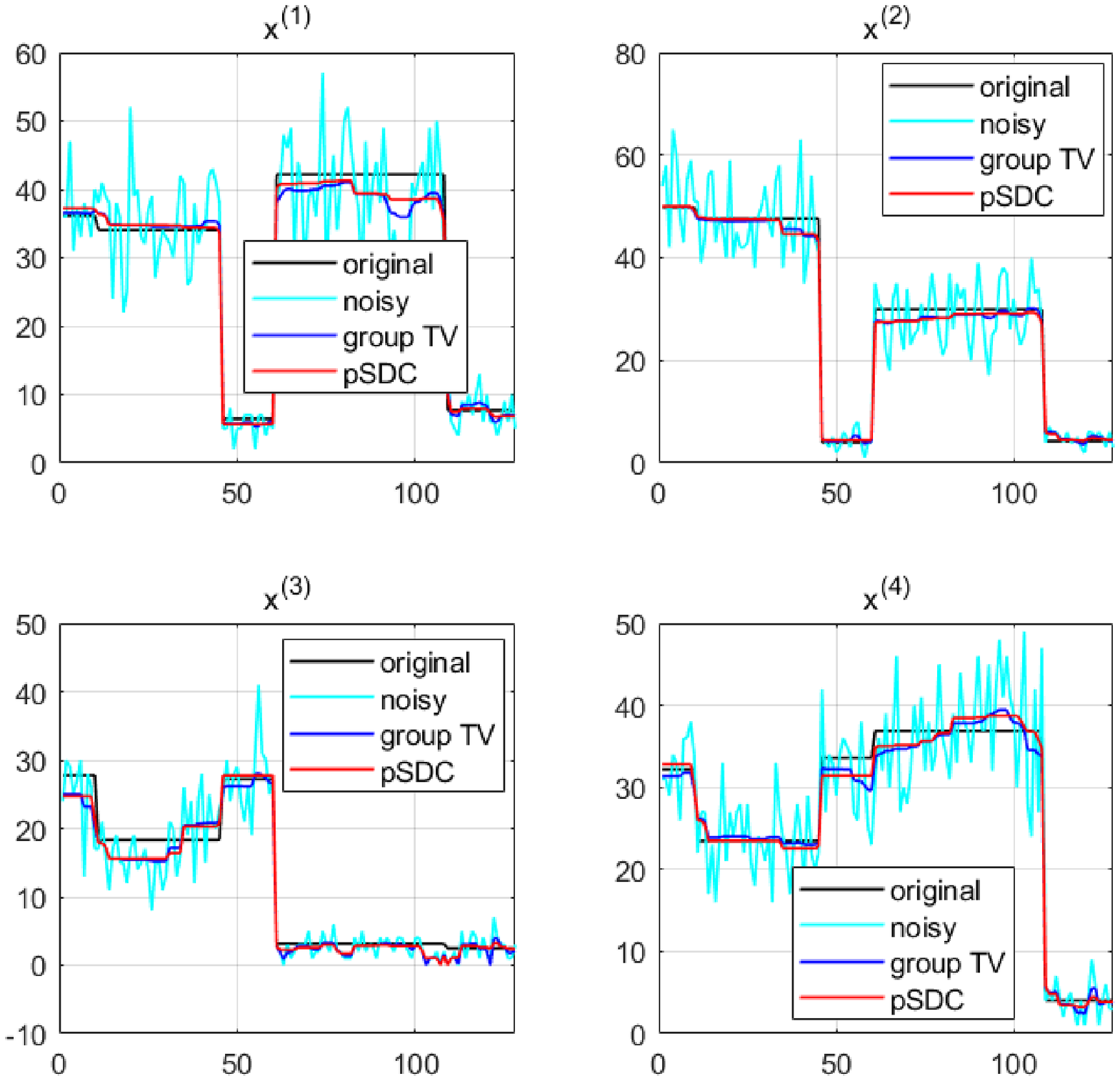}
		\caption{The recovered signals.}
		\label{fig:recovered_signals}
	\end{minipage}
\end{figure*}

\subsection{Group-Sparsity Based Poisson Denoising}\label{sec:poisson_denoising}
To demonstrate the potential of the proposed pSDC regularizer to be extended to more complicated scenarios, and to demonstrate the applicability of the proposed Algorithm \ref{alg:proposed_dca} to more general NRC models, we conduct numerical experiments in the problem of recovering jointly piecewise constant signals (i.e., all of the signals share the same change-points \cite{vert2010}) corrupted by Poisson noise. 

Let the group of piecewise constant signals to be estimated be $\Xm_\star=\left[\xv_\star^{(1)}, \xv_\star^{(2)}, \xv_\star^{(3)}, \xv_\star^{(4)}\right]\in\Real^{128\times 4}$, where $\xv_\star^{(1)}, \xv_\star^{(2)}, \xv_\star^{(3)}, \xv_\star^{(4)}$ are piecewise constant signals which share the same change-points. In each segment of $\xv_\star^{(i)}$ ($i=1,2,3,4$), the signal value is uniformly distributed on $[0,50]$, which implies that the true solution $\Xm_\star$ lies in the set:
\begin{equation}
C_X\coloneqq \braces{\Xm\in\Real^{128\times 4}\mid 0\leq x_{ji}\leq 50 \textrm{ for every }i,j},
\end{equation}
where $x_{ji}$ is the $(j,i)$ entry of $\Xm$. Consider the discrete differential operator $\Dm$ defined as
\begin{equation}
\Dm=\begin{bmatrix}
-1 & 1 &   & \\
 &  \ddots & \ddots &  \\
&  & -1 & 1 
\end{bmatrix}\in\Real^{127\times 128},
\end{equation}
then $\Dm\Xm_{\star}$ possesses group sparsity, i.e., $\Dm\Xm_{\star}$ has few rows with nonzero entries. Our goal is to recover $\Xm_{\star}$ from its noisy measurement $\Ym=\left[\yv^{(1)},\yv^{(2)},\yv^{(3)},\yv^{(4)}\right]\in\Real^{128\times 4}$, where the $j$th entry of $\yv^{(i)}$, say $y_{ji}$, follows the Poisson distribution
\begin{equation}
P\left(y_{ji}=k\mid x_{ji}\right)=\frac{x_{ji}^ke^{-x_{ji}}}{k!}, \quad k=0,1,2,\dots
\end{equation}
with $x_{ji}$ being the $j$th entry of $\xv^{(i)}_{\star}$. We note that this measurement model is also widely used in transmission computed tomography (TCT \cite{bouman1996}). With respect to this measurement model, a natural data fidelity term is the negative log-likelihood function $-\log(P(\Ym\mid\Xm))$, i.e.,
\begin{equation*}
F_{\textrm{TCT}}(\Xm)\coloneqq \begin{cases}
\sum_{j=1}^{128}\sum_{i=1}^4 \parentheses{x_{ji}-y_{ji}\log(x_{ji})}, & \textrm{if }x_{ji}>0,\\
+\infty, & \textrm{otherwise},
\end{cases}
\end{equation*}
which is proper, lower semicontinuous and convex.

We aim to estimate $\Xm_{\star}$ using the regularization technique described in this paper. We first present the design of the pSDC regularizer $\Psi_{p}(\cdot;\Pcal)$. Since $\Psi_p$ should promote the group sparsity of $\Dm\Xm$ and noting that the $l_{2,1}$-norm is the counterpart of the $l_1$-norm for promoting group sparsity \cite{eldar2010}, we adopt the following pSDC regularizer:
\begin{equation*}\label{eq:Psi_for_Poisson}
\Psi_{p}(\Xm;\Pcal)\coloneqq \norm{(\Dm\Xm)^T}_{2,1}-\parentheses{\norm{(\cdot)^T}_{2,1}\Box \phi_{\Pcal}}(\Dm\Xm),
\end{equation*}
where the smoothing function $\phi_{\Pcal}$ is defined as:
\begin{equation*}
\phi_{\Pcal}(\Zm)\coloneqq \sum_{i=1}^4 \frac{1}{2}\norm{\Bm_i\zv^{(i)}}^2_2
\end{equation*}
with $\Pcal\coloneqq (\Bm_i)_{i=1}^4$ and $\zv^{(i)}$ being the $i$th column vector of $\Zm\in\Real^{127\times 4}$. The regularization model we need to solve is 
\begin{equation}\label{eq:poisson_denoising}
\underset{\Xm\in C_X}{\mathrm{minimize}}\;\; F_{\mathrm{TCT}}(\Xm)+\lambda\Psi_{p}\parentheses{\Xm;(\Bm_i)_{i=1}^4}.
\end{equation}

Next we elaborate on the establishment of the overall-convexity condition of (\ref{eq:poisson_denoising}). According to Section \ref{subsec:NRC_model}, we can transform (\ref{eq:poisson_denoising}) into (\ref{eq:NRC_model}) by setting $\Xim\coloneqq \Dm$ and
\begin{align*}
F_1(\Xm) &\coloneqq F_{\textrm{TCT}}(\Xm)+\lambda\norm{(\Dm\Xm)^T}_{2,1}+\iota_{C_X}(\Xm),\\
F_2(\Zm) &\coloneqq \lambda\norm{\Zm^T}_{2,1},\\
\Phi_{\Pcal}(\Zm) &\coloneqq \frac{\lambda}{2}\sum_{i=1}^4 \norm{\Bm_i\zv^{(i)}}^2_2.
\end{align*}
Moreover, define $\xi_{ji}(z)\coloneqq z-y_{ji}\log(z)$, then we have
\begin{equation*}
F_{\textrm{TCT}}(\Xm)=\sum_{i=1}^4 \sum_{j=1}^{128} \xi_{ji}\left(\trace\parentheses{\Em_{ji}^T\Xm}\right),
\end{equation*}
where $\Em_{ji}\in\Real^{128\times 4}$ is the matrix with the $(j,i)$ entry being one and other entries being zero. One can verify that $F_1$ admits the following decomposition:
\begin{align*}
F_1(\Xm) &=F_1^s(\Xm)+F_1^n(\Xm),\textrm{ where }\\
F_1^s(\Xm) &\coloneqq F_{\textrm{TCT}}(\Xm)=\sum_{i=1}^4 \sum_{j=1}^{128} \xi_{ji}\left(\trace\parentheses{\Em_{ji}^T\Xm}\right),\\
F_1^n(\Xm) &\coloneqq \lambda\norm{(\Dm\Xm)^T}_{2,1}+\iota_{C_X}(\Xm).
\end{align*}
Define $C_1^s\coloneqq\braces{\Xm\in\Real^{128\times 4}\mid 0< x_{ji}< 50.01 \textrm{ for every }i,j}$, then $C_1^s\supset \dom{F_1}$ and one can verify that $\xi_{ji}$ is twice continuously differentiable on $\dom{\xi_{ji}}$ satisfying:
\begin{equation*}
(\forall \Xm\in C_1^s)\;\; \xi_{ji}''(\trace(\Em_{ji}^T\Xm))\geq \frac{y_{ji}}{50.01^2}\geq 0, \textrm{ for every }j,i.
\end{equation*}
In addition, note that $\Phi_{\Pcal}$ can be rewritten as
\begin{equation*}
\Phi_{\Pcal}(\Zm)\coloneqq \sum_{i=1}^4\sum_{j=1}^{128} \eta_{ji}(\trace(\bar{\Bm}_{ji}^T\Zm)),
\end{equation*}
where $\bar{\Bm}_{ji}\in\Real^{127\times 4}$ is the matrix which shares the same $j$th column as $\sqrt{\lambda}\Bm_i$ with the other columns being zeros, $\eta_{ji}(z)\coloneqq z^2\in\Gamma_0(\Real)$. One can verify that $\eta_{ji}$ is twice continuously differentiable on $\dom{\eta_{ji}}$ satisfying
\begin{equation*}
(\forall z\in\Real)\;\; 0\leq \eta_{ji}''(z)\leq 1,\textrm{ for every }j,i.
\end{equation*}
Therefore, we can apply Corollary \ref{corol:convexity_condition}-2) to the reformulation above, which yields the following overall-convexity condition for choosing $(\Bm_i)_{i=1}^4$:
\begin{equation}\label{eq:overall_convexity_condition_exp}
\frac{1}{50.01^2}\mathrm{diag}(y_{1,i},\dots,y_{128,i})\succeq \lambda\Dm^T\Bm_i^T\Bm_i\Dm.
\end{equation}
In our experiment we exploit the technique introduced in \cite[Prop. 2]{abe2020} to obtain such matrices satisfying (\ref{eq:overall_convexity_condition_exp}), wherein the values of $(\Bm_i)_{i=1}^4$ are determined by a hyperparameter $\theta\in[0,1]$. Especially, when $\theta=0$, $\Psi_p$ in (\ref{eq:poisson_denoising}) reduces to the group total variation (TV) regularizer $\norm{(\Dm\cdot)^T}_{2,1}$ \cite{selesnick2013}.

We solve the regularization model (\ref{eq:poisson_denoising}) by Algorithm \ref{alg:proposed_dca}. More precisely, we solve the subproblem (\ref{eq:alg1_step1}) by the forward-backward splitting algorithm \cite{combettes2011} and solve (\ref{eq:alg1_step2}) by the augmented Lagrangian method introduced in \cite{wu2011}. We change the values of $\lambda$ and $\theta\in[0,1]$, and evaluate the quality of each recovered signal $\Xm_{\textrm{est}}(\lambda,\theta)$ via the square error (SE):
\begin{equation*}
\mathrm{SE}\coloneqq \norm{\Xm_\mathrm{est}(\lambda,\theta)-\Xm_\star}^2_2.
\end{equation*}
Fig. \ref{fig:MSE_vs_lambda_poisson} shows the dependency of SE on $\lambda$ and $\theta$. From Fig. \ref{fig:MSE_vs_lambda_poisson}, one can verify that for all nonzero values of $\theta$, the proposed pSDC regularizer consistently defeats the group total variation regularizer. Fig. \ref{fig:recovered_signals} plots respectively the original signals (black), their noisy measurements (cyan) and recovered signals by group TV (blue) and the pSDC regularizer (red). From Fig. \ref{fig:recovered_signals}, one can see that the pSDC regularizer tends to produce recovered signals with fewer change points, hence leads to more accurate estimation. 

\section{Conclusion}\label{sec:conclusion}
In this paper, we have proposed a general class of convexity-preserving (CP) regularizers termed the pSDC regularizer. We have shown by presenting concrete examples that assigned with proper building blocks, the proposed pSDC regularizer reproduces existing CP regularizers and opens the way to a large number of new regularizers. With respect to the resultant nonconvexly regularized convex (NRC) model, we have derived a series of overall-convexity conditions which naturally embraces the conditions in previous works. We have proposed a globally convergent DC algorithm for solving the proposed NRC model. Numerical experiments have demonstrated the power of the proposed pSDC regularizer and the efficiency of the proposed DC algorithm.

\appendices

\section{Known Facts}\label{app:known_facts}
\begin{fact}[Sum of lower semicontinuous functions {\cite[Lemma 1.27]{bauschke2017}}]\label{fact:sum_of_ls_functions}
Let $(f_i)_{i\in I}$ be a finite family of lower semicontinuous functions from $\Real^n$ to $\Real\cup\braces{+\infty}$, and let $(\alpha_i)_{i\in I}$ be in $\Real_{++}$. Then $\sum_{i\in I}\alpha_i f_i$ is lower semicontinuous.
\end{fact}

\begin{fact}[Supremum of lower semicontinuous convex functions {\cite[Prop. 9.3]{bauschke2017}}]\label{fact:sup_of_lsc_functions}
Let $(f_i)_{i\in I}$ be a family of lower semicontinuous convex functions from $\Real^n$ to $\Real\cup \braces{+\infty}$, then $\sup_{i\in I}f_i$ is lower semicontinuous and convex.
\end{fact}

\begin{fact}[Some useful properties of subdifferential {\cite[Thm. 16.3 and Corol. 16.48]{bauschke2017}}]\label{fact:sum_of_subdifferential}
\hfill
\begin{enumerate}
\item Let $f:\Real^n\to\Real\cup\braces{+\infty}$ be proper. Then:
\begin{equation*}
\underset{\xv\in\Real^n}{\argmin}{f}=\braces{\xv\in\Real^n\mid \mathbf{0}_n\in\partial f(\xv)}.
\end{equation*}

\item Let $f,g\in\Gamma_0(\Real^n)$ and suppose that $\dom{g}=\Real^n$, then $\partial (f+g)=\partial f+\partial g$. 
\end{enumerate}

\end{fact}

\begin{fact}[Twice continuously differentiable convex functions {\cite[Thm. 2.1.4]{nesterov2018}}]\label{fact:twice_differentiable_convex}
Let $Q\subset\Real^n$ be an open set. A twice continuously differentiable function $f:\Real^n\to\Real$ is convex on $Q$ if and only if for every $\xv\in Q$, $\nabla^2 f(\xv)\succeq \mathbf{O}_{n\times n}$.
\end{fact}

\begin{fact}[Coercive proper lower semicontinuous functions {\cite[Prop. 11.12 and 11.13]{bauschke2017}}]\label{fact:coercive_lsc_functions}
Let $f\in\Gamma_0(\Real^n)$. Then:
\begin{enumerate}
\item $f$ is coercive if and only if for every $\xi\in\Real$, the level set $\levelset_{\leq \xi}f$ is bounded.

\item $f$ is coercive if and only if there exists $\xi\in\Real$ such that $\levelset_{\leq \xi}f$ is nonempty and bounded.
\end{enumerate}
\end{fact}

\begin{fact}[Properties of the simplified DCA {\cite[Thm. 3 (i)(iv)]{pham1997}}]\label{fact:properties_of_DCA} Let $(\xv_k)_{k\in\Natural}$ and $(\uv_k)_{k\in\Natural}$ be sequences defined by Algorithm \ref{alg:simplified_dca} for (\ref{eq:DC_program}). Then the following holds:
\[J(\xv_{k+1})\leq J(\xv_k),\]
Moreover, suppose that $J$ in (\ref{eq:DC_program}) is bounded below, i.e., \[\inf_{\xv\in\Real^n}J(\xv)>-\infty,\]
and suppose that $(\xv_k)_{k\in\Natural}$ and $(\uv_k)_{k\in\Natural}$ are bounded. Then for every limit point $\bar{\xv}$ of $(\xv_k)_{k\in\Natural}$, the following holds:
\begin{enumerate}
\item there exists a limit point $\bar{\uv}$ of $(\uv_k)_{k\in\Natural}$ such that $\bar{\uv}\in\partial g(\bar{\xv})\cap \partial h(\bar{\xv})$,

\item $\lim_{k\to +\infty}J(\xv_k)=J(\bar{\xv})$.
\end{enumerate}

\end{fact}

\section{Proof of Theorem \ref{thm:general_convexity_condition}}
\label{app:general_convexity_condition}

\begin{proof}
We first show that $J_{p}(\cdot;\mathcal{P}_o)$ is proper. Since $F_1\in\Gamma_0(\Real^n)$ from assumption, $F_1$ is proper, lower semicontinuous and convex. Hence there exists $\xv_0\in \dom{F_1}$, and we have
\begin{flalign*}
J_{p}(\xv_0;\mathcal{P}_o)= F_1(\xv_0)-(F_2\Box\Phi_{\Pcal_o})(\Xim\xv_0)< +\infty
\end{flalign*}
from Assumption \ref{assump:finite_inf_conv}. Hence $J_{p}(\cdot;\Pcal_o)$ is proper.

Next we prove that $J_{p}(\cdot;\mathcal{P}_o)$ is lower semicontinuous and convex. Since $\Phi_{\Pcal_o}(\cdot)$ is continuous, $-\Phi_{\Pcal_o}(\Xim\cdot-\zv)$ is continuous (thus is lower semicontinuous). Since $F_1$ is lower semicontinuous, Fact \ref{fact:sum_of_ls_functions} in Appendix \ref{app:known_facts} implies that
\[M_{\zv}(\cdot;\Pcal_o)\coloneqq F_1(\cdot)-\Phi_{\Pcal_o}(\Xim\cdot-\zv)\]
is lower semicontinuous for every $\zv\in\Real^q$. 

We rewrite $J_{p}(\xv;\Pcal_o)$ as follows,
\begin{eqnarray*}
J_{p}(\xv;\mathcal{P}_o)& = & F_1(\xv)-(F_2\Box\Phi_{\mathcal{P}_o})(\bm{\Xi} \xv)\\
 &= & F_1(\xv)-\inf_{\zv\in\Real^q}\braces{F_2(\zv)+\Phi_{\mathcal{P}_o}(\bm{\Xi}\xv-\zv)}\\
&=& F_1(\xv)-\inf_{\zv\in\dom{F_2}}\braces{F_2(\zv)+\Phi_{\mathcal{P}_o}(\bm{\Xi}\xv-\zv)}\\
&=& \sup_{\zv\in\dom{F_2}}\braces{F_1(\xv)-F_2(\zv)-\Phi_{\mathcal{P}_o}(\bm{\Xi}\xv-\zv)}\\
&=& \sup_{\zv\in\dom{F_2}} \braces{M_{\zv}(\xv;\mathcal{P}_o)-F_2(\zv)}.
\end{eqnarray*}
Accordingly, $J_{p}(\cdot;\Pcal_o)$ is the supremum of the following family of functions:
\[\parentheses{M_{\zv}(\cdot;\Pcal_o)-F_2(\zv)}_{\zv\in\dom{F_2}},\]
Since we have proved that $M_{\zv}(\cdot;\Pcal_o)$ is lower semicontinuous for every $\zv\in\Real^q$ ,and $M_{\zv}(\cdot;\Pcal_o)$ is convex for every $\zv\in\dom{F_2}$ from assumption in Theorem \ref{thm:general_convexity_condition},
$J_{p}(\cdot;\Pcal_o)$ is the supremum of a family of lower semicontinuous convex functions, hence is lower semicontinuous convex (Fact \ref{fact:sup_of_lsc_functions}). 

Combining the discussion above completes the proof.
\end{proof}

\section{Proof of Corollary \ref{corol:convexity_condition}}
\label{app:specific_convexity_condition}
\begin{proof}
1) For a fixed $\zv\in\Real^q$, we define
\begin{equation*}
\bar{M}_{\zv}(\xv;\Pcal_o)\coloneqq F_{1}^s(\xv)-\Phi_{\Pcal_o}(\Xim \xv-\zv).
\end{equation*}
Taking the Hessian of $\bar{M}_{\zv}(\xv;\Pcal_o)$ on $C_1^s$ yields (cf. \cite[A.4.4]{boyd2004} for derivation of $\nabla^2  \bar{M}_{\zv}(\xv;\Pcal_o)$)
\begin{equation*}
\nabla^2 \bar{M}_{\zv}(\xv;\Pcal_o)=\nabla^2 F_{1}^s(\xv)-\Xim^T\nabla^2 \Phi_{\Pcal_o}(\Xim \xv-\zv)\Xim.
\end{equation*}
The following holds from assumption:
\begin{equation*}
(\forall \xv\in C_1^s,\forall \zv\in\dom{F_2})\;\; \nabla^2 \bar{M}_{\zv}(\xv;\Pcal_o)\succeq \mathbf{O}_{n\times n}.
\end{equation*}
Hence Fact \ref{fact:twice_differentiable_convex} guarantees that $
\nabla^2 \bar{M}_{\zv}(\cdot;\Pcal_o)$ is convex on $C_1^s$ for every $\zv\in\dom{F_2}$. Moreover, since
\begin{align*}
M_{\zv}(\cdot;\Pcal_o) &\coloneqq F_1(\cdot)-\Phi_{\Pcal_o}(\Xim\cdot-\zv)\\
&=F_1^s(\cdot)+F_1^n(\cdot)-\Phi_{\Pcal_o}(\Xim\cdot-\zv)  \\
&=\bar{M}_{\zv}(\cdot;\Pcal_o)+F_1^n(\cdot),
\end{align*}
for every $\zv\in\dom{F_2}$, the convexity of $F_1^n$ and the convexity of $\bar{M}_{\zv}(\cdot;\Pcal_o)$ on $C_1^s\supset\dom{F_1}$ yields the convexity of $M_{\zv}(\cdot;\Pcal_o)$. Thus $J_{p}(\cdot;\Pcal_o)\in\Gamma_0(\Real^n)$ from Theorem \ref{thm:general_convexity_condition}. 

\noindent
2) From the assumption, we have the following:
\begin{equation*}
\nabla^2 F_1^s(\xv)=\Am^T \begin{bmatrix}
\xi_1''(\av_1^T\xv) & & \\
 &  \ddots & \\
 & & \xi_m''(\av_m^T\xv)
\end{bmatrix}
\Am.
\end{equation*}
Since $\xi_i''(\av_i^T\xv)\geq \gamma_i$ for every $\xv\in C_1^s$, the following holds
\begin{equation*}
\begin{bmatrix}
\xi_1''(\av_1^T\xv) & &  \\
 & \ddots & \\
 & & \xi_m''(\av_m^T\xv)
\end{bmatrix}\succeq
\mathrm{diag}\parentheses{\gamma_1,\dots,\gamma_m}
\end{equation*}
for every $\xv\in C_1^s$, which implies the following inequality:
\begin{equation}\label{eq:hessian_F1s_geq_AtA}
(\forall \xv\in C_1^s)\;\;\nabla^2 F_1^s(\xv)\succeq \Am^T\mathrm{diag}\parentheses{\gamma_1,\dots,\gamma_m}\Am.
\end{equation}
On the other hand, consider the smoothing function \[\Phi_{\Pcal}(\zv)\coloneqq \sum_{j=1}^p\eta_j(\bv_j^T\zv),\]
then for every $\zv\in\Real^q$, the following holds:
\begin{equation*}
\nabla^2\Phi_{\Pcal}(\zv)=\Bm^T \begin{bmatrix}
\eta_1''(\bv_1^T\zv) & & \\
 & \ddots & \\
 & & \eta_p''(\bv_p^T\zv)
\end{bmatrix}
\Bm.
\end{equation*} 
Since $\eta''_j(z)\leq \kappa_j$ for every $z\in\Real$, we have
\begin{equation*}
\begin{bmatrix}
\eta_1''(\bv_1^T\zv) & & \\
 & \ddots & \\
 & & \eta_p''(\bv_p^T\zv)
\end{bmatrix}\preceq \mathrm{diag}\parentheses{\kappa_1,\dots, \kappa_p},
\end{equation*}
which implies the following inequality:
\begin{equation*}
(\forall \zv\in\Real^q)\;\; \nabla^2 \Phi_{\Pcal}(\zv)\preceq \Bm^T\mathrm{diag}\parentheses{\kappa_1,\dots,\kappa_p}\Bm.
\end{equation*}
From the assumption that $\Pcal_o\coloneqq\Bm_o$ satisfies
\begin{equation*}
\Am^T\mathrm{diag}\parentheses{\gamma_1,\dots,\gamma_m}\Am\succeq\Xim^T\Bm_o^T\mathrm{diag}\parentheses{\kappa_1,\dots,\kappa_p}\Bm_o\Xim,
\end{equation*}
we conclude that for every $\zv\in\Real^q$,
\begin{equation*}
\Am^T\mathrm{diag}\parentheses{\gamma_1,\dots,\gamma_m}\Am\succeq \Xim^T\nabla^2\Phi_{\Pcal_o}(\zv)\Xim,
\end{equation*}
combining which with (\ref{eq:hessian_F1s_geq_AtA}) yields the following
\begin{equation*}
(\forall\xv\in C_1^s,\forall \zv\in \Real^q)\;\;\nabla^2 F_1^s(\xv)\succeq \Xim^T\nabla^2\Phi_{\Pcal_o}(\zv)\Xim.
\end{equation*} 
Hence $J_{p}(\cdot;\Pcal_o)\in\Gamma_0(\Real^n)$ from the result of 1). 
\qedhere
\end{proof}

\section{Proof of Theorem \ref{thm:convergence_of_dc_algorithm}}
\label{app:convergence_of_DC_algorithm}

\begin{proof}
1) The result follows from Fact \ref{fact:properties_of_DCA}.

\noindent
2) Since the overall-convexity condition holds, we have $J_{p}(\cdot;\Pcal_o)\in\Gamma_0(\Real^n)$ by Theorem \ref{thm:general_convexity_condition}. Define
\begin{equation*}
\alpha\coloneqq \inf_{\xv\in\Real^n}J_{p}(\xv;\Pcal_o),
\end{equation*}
then from Assumption \ref{assump:boundedness_of_argmin},
\[\levelset_{\leq \alpha}\;J_{p}(\cdot;\Pcal_o)=\underset{{\xv\in\Real^n}}{\argmin}\;J_{p}(\cdot;\Pcal_o)\]
is nonempty and bounded. By Fact \ref{fact:coercive_lsc_functions}, $J_{p}(\cdot;\Pcal_o)$ is coercive, which implies that for $\xi_0\coloneqq J_{p}(\xv_0;\Pcal_o)$, $\levelset_{\leq \xi_0}\; J_{p}(\cdot;\Pcal_o)$
is bounded. According to the result 1), for every $k\geq 1$,
\[J_{p}(\xv_k;\Pcal_o)\leq J_{p}(\xv_0;\Pcal_o)=\xi_0,\]
hence $(\xv_k)_{k\in\Natural}\subset \levelset_{\leq \xi_0}\; J_{p}(\cdot;\Pcal_o)$ is bounded. 

Since we have $\uv_k=\nabla \bar{F}_2(\xv_k;\Pcal_o)$ and $\nabla \bar{F}_2(\cdot;\Pcal_o)$ is continuous from Lemma \ref{lemma:nabla_bar_F2}, we can yield the boundedness of of $(\uv_k)_{k\in\Natural}$ from the boundedness of $(\xv_k)_{k\in\Natural}$.

\noindent
3) Let $\bar{\xv}$ be a limit point of $(\xv_k)_{k\in\Natural}$. By Fact \ref{fact:properties_of_DCA}, there exists a limit point $\bar{\uv}$ of $(\uv_k)_{k\in\Natural}$ such that
\begin{equation*}
\bar{\uv}\in\partial F_1(\bar{\xv})\cap \partial \bar{F}_2(\bar{\xv};\Pcal_o)=\partial F_1(\bar{\xv})\cap \braces{\nabla \bar{F}_2(\bar{\xv};\Pcal_o)}.
\end{equation*}
The inclusion above implies that $\bar{\uv}=\nabla \bar{F}_2(\bar{\xv};\Pcal_o)$ and
\begin{equation*}
\nabla \bar{F}_2(\bar{\xv};\Pcal_o)\in\partial F_1(\bar{\xv}).
\end{equation*}
Hence the following inclusion hold
\begin{equation}\label{eq:zero_inclusion}
\mathbf{0}_n\in \braces{\uv-\nabla \bar{F}_2(\bar{\xv};\Pcal_o) \mid \uv\in\partial F_1(\bar{\xv})}.
\end{equation}

Since $J_{p}(\bar{\xv};\Pcal_o)=F_1(\bar{\xv})-\bar{F}_2(\bar{\xv};\Pcal_o)$, we have 
\begin{equation*}
F_1(\bar{\xv})=J_{p}(\bar{\xv};\Pcal_o)+\bar{F}_2(\bar{\xv};\Pcal_o).
\end{equation*}
By Thm. \ref{thm:general_convexity_condition} and Lemma \ref{lemma:nabla_bar_F2}, $J_{p}(\cdot;\Pcal_o)$ and $\bar{F}_2(\cdot;\Pcal_o)$ are in $\Gamma_0(\Real^n)$. In addition, $\dom{\bar{F}_2(\cdot;\Pcal_o)}=\Real^n$ by Assumption \ref{assump:finite_inf_conv}. Hence Fact \ref{fact:sum_of_subdifferential} guarantees the following equality:
\begin{equation*}
\partial F_1(\bar{\xv})=\partial J_{p}(\bar{\xv};\Pcal_o)+\partial \bar{F}_2(\bar{\xv};\Pcal_o),
\end{equation*}
The single-valuedness of $\partial \bar{F}_2(\bar{\xv};\Pcal_o)$ then yields
\begin{equation*}
\partial F_1(\bar{\xv})=\braces{\vv+\nabla \bar{F}_2(\bar{\xv};\Pcal_o) \mid \vv\in\partial J_{p}(\bar{\xv};\Pcal_o)},
\end{equation*}
which implies that
\begin{equation*}
\braces{\uv-\nabla\bar{F}_2(\bar{\xv};\Pcal_o)\mid \uv\in\partial F_1(\bar{\xv})}=\partial J_{p}(\bar{\xv};\Pcal_o).
\end{equation*}
Substituting the equation above into (\ref{eq:zero_inclusion}) yields that
\begin{equation*}
\mathbf{0}_n\in\partial J_{p}(\bar{\xv};\Pcal_o).
\end{equation*}
Since $J_{p}(\cdot;\Pcal_o)\in\Gamma_0(\Real^n)$, by Fact \ref{fact:sum_of_subdifferential} we have
\begin{equation*}
J_{p}(\bar{\xv};\Pcal_o)=\alpha\coloneqq \min_{\xv\in\Real^n} J_{p}(\xv;\Pcal_o)
\end{equation*}

\noindent
4) Since $(\xv_k)_{k\in\Natural}$ is bounded from 2), there exists a subsequence of $(\xv_{k})_{k\in\Natural}$ converging to some point $\bar{\xv}\in\Real^n$. Hence $J_{p}(\bar{\xv};\Pcal_o)=\alpha$ from the result 3). Fact \ref{fact:properties_of_DCA} then yields
\begin{equation*}
\lim_{k\to +\infty} J_{p}(\xv_k;\Pcal_o)=J_{p}(\bar{\xv};\Pcal_o)=\alpha. \qedhere
\end{equation*}
\end{proof}

\bibliographystyle{IEEEtran}
\bibliography{refs}

\end{document}